\documentclass[twocolumn,review]{autart_arxiv}

\usepackage{graphicx}
\usepackage{amsmath}
\usepackage{amssymb}
\usepackage{algorithm,algorithmic}
\usepackage[center]{subfigure}
\usepackage{euscript}
\usepackage{verbatim}
\usepackage{amssymb}
\usepackage{cite}
\usepackage{color}
\usepackage{cancel}
\usepackage{calligra}
\usepackage{pgf}
\usepackage{tikz}
\usepackage{ctable}
\usetikzlibrary{arrows,automata}
\usetikzlibrary{shapes}
\usetikzlibrary{positioning}
\usepackage{bbm}
\usepackage{multirow}
\usepackage{url}

\newtheorem{definition}{\textbf{Definition}}{\normalfont}{\normalfont}
{\normalfont}{\normalfont}
\newtheorem{remark}{Remark}{\normalfont}{\normalfont}
\newtheorem{theorem}{Theorem}

\newtheorem{lemma}{Lemma}
\newenvironment{proof}{\emph{Proof:}}{\hfill$\square$}

\renewcommand{\figurename}{Figure}
\renewcommand{\tablename}{Table}

\usepackage{enumitem}
\renewcommand{\theenumi}{\arabic{enumi}}

\renewcommand{\theenumii}{\arabic{enumii}}

\renewcommand{\theenumiii}{\arabic{enumiii}}

\begin{document}

	\begin{frontmatter}
		\runtitle{Learning explicit predictive controllers}
		\title{Learning explicit predictive controllers:\\theory and applications\thanksref{funding}
		}
		
		\thanks[funding]{This project was partially supported by the Italian Ministry of University and Research under the PRIN'17 project \textquotedblleft Data-driven learning of constrained control systems", contract no. 2017J89ARP.}
		
		\author[Polimi]{Andrea Sassella}\ead{andrea.sassella@polimi.it},
		\author[Polimi]{Valentina Breschi}\ead{valentina.breschi@polimi.it},
		\author[Polimi]{Simone Formentin}\ead{simone.formentin@polimi.it}
		
		\address[Polimi]{Dipartimento di Elettronica, Informazione e Bioingegneria, Politecnico di Milano, Piazza L. da Vinci 32, 20133 Milano, Italy.}
		

		\begin{abstract}
			In this paper, we deal with data-driven predictive control of linear time-invariant (LTI) systems. Specifically, we show for the first time how \textit{explicit} predictive laws can be learnt directly from data, without needing to identify the system to control. To this aim, we resort to the Willems' fundamental lemma and we derive the explicit formulas by suitably elaborating the constrained optimization problem under investigation. The resulting optimal controller turns out to be a piecewise affine system coinciding with the solution of the original model-based problem in case of noiseless data. Such an equivalence is proven to hold asymptotically also in presence of measurement noise, thus making the proposed method a computationally efficient (but model-free) alternative to the state of the art predictive controls. The above statements are further supported by numerical simulations on three benchmark examples.
		\end{abstract}
		
	\end{frontmatter}
	
	\section{Introduction} 
	Model predictive control (MPC) is a widely recognized and diffused technology for the solution of advanced constrained control (see \emph{e.g.,} \cite{Rawlings2000,Hrovat2012}), yet still subject of active research \cite{Faulwasser2021,Matne2014}. MPC relies on $(i)$ the formulation of a finite-horizon open-loop optimal control problem, that is iteratively solved in a \emph{receding-horizon} fashion to determine the optimal control move at each time instant, and on $(ii)$ the availability of a model describing the dynamics of the system under control. 
	
	Nonetheless, solving a mathematical problem in real-time might be infeasible or not advisable for certain applications, mainly due to the cost of the computational equipment needed to retrieve the optimal action within the sampling time or because of software certification concerns, that might arise in safety critical applications \cite{Alessio2009}. To cope with these limitations, over the years several techniques have been proposed to improve the efficiency of MPC solvers \cite{Patrinos2014} or to provide complexity certifications of the latter \cite{Cimini2017,Cimini2021}. As an alternative to these solutions, \cite{Berberich2020b} propose to move the design effort off-line, by exploiting the structure of the model predictive control problem to derive an \emph{explicit} MPC law. Although this seminal work focused on quadratic costs and linear dynamic models, since then explicit MPC has been extended to several classes of objectives, \emph{e.g.,} linear costs \cite{Bemporad2002}, and to more complex models (see \cite{Alessio2009} for a complete overview). Notably, when exploiting explicit MPC, the computation of the control action entails a function evaluation, thus not requiring sophisticated and potentially demanding optimization procedures. Meanwhile, the complexity of the explicit control law can become unmanageable for large scale problems or when long prediction horizons are considered, thus making this approach mainly suited for relatively small control problems.
	
	\begin{table*}[!tb]
		\caption{Requirements of E-DDPC as compared to implicit data-driven predictive control (DDPC), implicit and explicit MPC. The crosses (\texttt{x}) indicate the elements needed for the design and deployment of the controller.}\label{tab:comparison}
		\centering
		\begin{tabular}{ccccc}
			\multicolumn{1}{c}{} \hspace*{-.3cm}&\hspace*{-.3cm} Implicit MPC \hspace*{-.15cm}&\hspace*{-.15cm} Explicit MPC \hspace*{-.145cm}&\hspace*{-.145cm} DDPC \hspace*{-.15cm}&\hspace*{-.15cm} E-DDPC\\ 
			\hline
			Dataset &\texttt{x} & \texttt{x} & \texttt{x} &\texttt{x}\\
			\hline 
			Identification & \texttt{x} & \texttt{x} & - & -\\
			\hline 
			Online solver & \texttt{x} & - & \texttt{x} & -\\
			\hline 
		\end{tabular}
	\end{table*}
	
	The need for a model of the controlled system can also be a shortcoming of standard MPC, especially for those applications in which a mathematical model for the plant is not available and it has to be retrieved from data. To explain the dynamics of a system, several identification techniques have been proposed over the years (see \cite{Ljung1999} for an overview on classical identification techniques) to learn dynamical models from data and learning-based predictive control methods have been devised to handle possible inaccuracies in these data-driven models, \emph{e.g.,} \cite{ASWANI2013}. Nonetheless, identification procedures are known to be generally expensive and time consuming. Meanwhile, identification techniques usually aim at achieving the maximum model accuracy, often at the price of overly complex model structures, while seldom accounting for the application these models are learned for \cite{GEVERS2005}. Indeed, in some cases the added complexity introduced by the identified model is unnecessary to achieve a given control goal. To overcome these limitations, different data-driven control approaches have been proposed to skip the identification phase and to directly learn the controller from data, ranging from model-reference methods \cite{Hjalmarsson2002,Breschi2021,formentin2014comparison,formentin2019deterministic} to approaches for the design of linear quadratic regulators \cite{DePersis2021} and state feedback controllers \cite{Berberich2020,Rotulo2021online}. Along this research line, several data-based predictive strategies have been recently proposed, which ground on results in behavioral theory to formulate purely data-driven control predictive problems \cite{coulson2019data,Berberich2021}. These foundational approaches have then been extended to handle tracking problems \cite{Berberich2020b}, to deal with nonlinear systems \cite{Berberich2021b}, and to exploit regularization to improve the performance of the final controller \cite{Dorfler2021,Coulson2019}.
	
	Along this line, in this paper we combine for the first time the benefits of explicit MPC and the ones of purely data-driven methods into a fully \emph{explicit data-driven predictive control} (E-DDPC) approach. The proposed procedure leads to the definition of a purely data-based piecewise affine (PWA) control law, which is defined so as to guarantee constraint satisfaction and the optimization of a quadratic performance-oriented cost. A key condition for its derivation is the persistency of excitation of the input signal, that allows us to exploit Willems et al.'s lemma \cite{Willems2005} to translate the explicit MPC solution into its data-driven counterpart. The proposed derivation also allows us to prove the equivalence between the model-based and the data-based solutions in case of noiseless measurements, while we present a noise handling strategy to obtain asymptotic equivalence in presence of measurement noise. 
	As summarized in \tablename{~\ref{tab:comparison}}, the presented explicit predictive controller only relies on a single dataset to be designed, whereas complex identification procedures or online solvers are no longer required.  
	
	The paper is organized as follows. In Section~\ref{sec:problem}, we lay out the main assumptions on the data and we introduce the control problem of interest. The standard model-based predictive control problem and the main steps leading to the formulation of its explicit solution are summarized in Section~\ref{sec:fromItoE}. Section~\ref{sec:E-DDPC} introduces the main theoretical results, leading to the definition of the data-driven explicit predictive controller. The implementation of the data-based control law is then illustrated in Section~\ref{sec:practice}, along with the proposed practical procedure to handle noisy data. Section~\ref{sec:examples} shows the results obtained using the data-driven explicit predictive controller on three simulation case studies, \emph{i.e.,} the stabilization of the benchmark system considered in \cite{Bemporad2002b}, the regulation to zero of a sparse system and the altitude control of a quadrotor. The paper is ended by some concluding remarks.
	
	\subsection*{Notation} 
	Let $\mathbb{N}_{0}$ and $\mathbb{R}$ be the set of natural numbers, including zero, and the set of real numbers, respectively. Denote with $\mathbb{R}^{n}$ and $\mathbb{R}^{n \times m}$ the set of real column vector of dimension $n$ and the set of real matrices of dimension $n \times m$, respectively. Given a rectangular matrix $B \in \mathbb{R}^{m \times n}$, $B' \in \mathbb{R}^{n \times m}$ denotes its transpose and $B^{\dagger}$ indicates its right inverse. Given a squared matrix $A \in \mathbb{R}^{n\times n}$, we define its inverse as $A^{-1}$. We denote with $I_{n}$ the identity matrix of dimension $n$ and with $\mathbf{0}_{n \times m}$ a zero matrix of dimension $n \times m$. If $Q \succ 0$ ($Q \succeq 0$), then the matrix $Q \in \mathbb{R}^{n \times n}$ is  positive definite (semi-positive definite). Let $x \in \mathbb{R}^{n}$, the quadratic form $x'Qx$ is compactly denoted as $\|x\|_{Q}^{2}$.
	
	\section{Problem formulation}\label{sec:problem}
	Consider a \emph{linear time-invariant} (LTI) system, whose dynamics is described by the following \emph{unknown} state-space model:
	\begin{subequations}\label{eq:system}
		\begin{equation}
			\mathcal{S}: ~~\begin{cases}
				x(t+1)=Ax(t)+Bu(t),\\
				y(t)=Cx(t)+Du(t),
			\end{cases}
		\end{equation}
		where $x(t) \in \mathbb{R}^{n}$ is the state of the system at time $t \in \mathbb{N}_{0}$, $u(t) \in \mathbb{R}^{m}$ is an exogenous input and $y \in \mathbb{R}^{p}$ is the corresponding \emph{noiseless} output. Let us assume that the unknown system $\mathcal{S}$ is controllable and that its state is \emph{fully measurable}, namely 
		\begin{equation}
			y(t)=x(t),~~\forall t \in \mathbb{N}_{0},
		\end{equation}	
		\emph{i.e.,} 
		$C=I_{n}$ and $D=\mathbf{0}_{n \times m}$. 
	\end{subequations}
	
	Assume that we can excite the system with an input sequence $\mathcal{U}_{T}=\{u(t)\}_{t=0}^{T}$, that is \emph{persistently exciting} of order $n+1$ according to the following definition.
	\begin{definition}[Persistently exciting input \cite{Willems2005}] \label{def:persistentlyexciting}
		The input sequence $\mathcal{U}_{T}$ is said to be persistently exciting of order $\tau$ if the matrix
		\begin{equation} U_{0,\tau,T} = \begin{bmatrix}u(0) & u(1) & \cdots & u(T-\tau)\\
				u(1) & u(2) & \cdots & u(T-\tau+1)\\
				\vdots & \vdots & \cdots & \vdots \\
				u(\tau-1) & u(\tau) & \cdots & u(T-1)
			\end{bmatrix} \in \mathbb{R}^{m\tau \times T} \end{equation}
		is full row rank, namely $\mathrm{rank}(U_{0,\tau,T})=m\tau$.
	\end{definition}
	Suppose that we can only measure the corresponding \emph{noisy} output sequence $\mathcal{Y}_{T}^{n}=\{y^{n}(t)\}_{t=0}^{T}$, with
	\begin{equation}\label{eq:noisy_y}
		y^{n}(t)=y(t)+v(t),
	\end{equation} 
	where $v$ is a zero-mean white noise with covariance matrix $\Sigma \in \mathbb{R}^{n \times n}$. 
	
	In this work, our goal is to exploit the available data $\mathcal{D}_{T}^{n}=\{\mathcal{U}_{T},\mathcal{Y}_{T}^{n}\}$ and their features to \emph{directly} and \emph{explicitly} solve standard predictive control problems without first identifying a model of $\mathcal{S}$. More specifically, let $N_{x},N_{c},N_{u}$, respectively denote the state, input and constraint horizons, with $N_{u} \leq N_{x}$. Our objective is to solve the following model-based constrained optimal control problem:
	\begin{subequations}\label{eq:MPC}
		\begin{align} 
			\min_{\{u(k)\}_{k=0}^{N_{u}-1}} &~~ \|x(N_{x})\|_{P}^{2}+\!\!\!\sum_{k=0}^{N_{x}-1}\left[\|x(k)\|_{Q}^{2}\!+\!\|u(k)\|_{R}^{2}\right]\label{eq:MPCcost}\\
			\mbox{s.t.} &~x(k+1)\!=\!Ax(k)+Bu(k),~k\geq 0, \label{eq:MPCconstr1}\\
			&~x(0)=x, \label{eq:MPCconstr2}\\
			&~\mathcal{C}_{x}x(k)\!+\!\mathcal{C}_{u}u(k) \leq d,~ k\!=\!0,\ldots,N_{c}\!-\!1, \label{eq:MPCconstr3}\\
			&~u(k)=Kx(k),~N_{u} \leq k < N_{x}, \label{eq:MPCconstr4}
		\end{align}
		which aims at finding the optimal sequence of inputs $\{u(k)\}_{k=0}^{N_{u}}$ steering the state of the system to zero from the initial condition in \eqref{eq:MPCconstr3}, while satisfying the \emph{convex} constraints in \eqref{eq:MPCconstr3}, without identifying a model for the system under control. Note that, at each time instant $t$, the initialization in \eqref{eq:MPCconstr2} relies on the latest available information on the state, \emph{i.e.,} its measurement ($x=x(t)$) or estimate ($x=\hat{x}(t)$). In this formulation, optimality is thus dictated by: $(i)$ the distance of the predicted state from the origin, weighted by $Q \succeq 0$; $(ii)$ the control effort, penalized with $R \succ 0$, and $(iii)$ a terminal cost, with associated penalty $P \succeq 0$. Moreover, whenever the state and input horizon are different, the constraint in \eqref{eq:MPCconstr4} entails that some precomputed feedback gain $K \in \mathcal{R}^{m \times n}$ is used to generate the input for $k=N_{u},\ldots,N_{x}-1$.   
	\end{subequations}
	
	\section{From implicit to explicit MPC: an overview}\label{sec:fromItoE}
	In this section, we summarize the main steps required to shift from the standard implicit MPC formulation to the explicit solution of a predictive control problem within a model-based setting. A similar procedure will be key in the subsequent derivation of the fully \emph{data-driven} explicit predictive control solution.   
	
	To start with, we notice that, by rewriting the prediction model in \eqref{eq:MPCconstr1} as 
	\begin{equation}
		x(k)=A^{k}x+\sum_{j=0}^{k-1}A^{j}Bu(k-j-1), ~~ k\geq 0,
	\end{equation}
	and accounting for the fact that $u(k)$ is dictated by the feedback law in \eqref{eq:MPCconstr4} whenever $k \geq N_{u}$, the control problem in \eqref{eq:MPC} can be recast as the following convex \emph{multi-parametric quadratic program} (mp-QP):
	\begin{subequations}\label{eq:QP}
		\begin{align} \underset{U}{\text{min}}\quad&U'HU+2x'FU, \label{eq:QP_cost}\\ 
			\quad \text{s.t.}\quad & GU\leq W+Ex, \label{eq:QP_constr} 
		\end{align}
	\end{subequations}
	where $U \in \mathbb{R}^{N_{u}m}$ stacks the sequence of control actions to be optimized, and the matrices $H$, $F$, $G$, $W$ and $E$ are functions of the penalties $Q$, $R$, $P$, the precomputed feedback gain $K$ and the matrices characterizing the evolution of the system in \eqref{eq:system}. By completing the squares, the mp-QP in \eqref{eq:QP} can be equivalently reformulated as
	\begin{subequations} \label{eq:QP2}
		\begin{align} \underset{z}{\text{min}}\quad&z'Hz, \label{eq:QP2_cost}\\ 
			\quad \text{s.t.}\quad & Gz\leq W+Sx, \label{eq:QP2_constr}
		\end{align}
	\end{subequations}
	where $H \succ 0$, $S\triangleq E+GH^{-1}F'$ and the new optimization variable is
	\begin{equation}\label{eq:Zexpression}
		z \triangleq U+H^{-1}F'x\in\mathbb{R}^{N_{u}m}.
	\end{equation}
	
	Since $H$ in \eqref{eq:QP2} is positive definite, the solution of the mp-QP in \eqref{eq:QP2} is unique and it can be retrieved in closed-form from the Karush-Kuhn-Tucker (KKT) optimality conditions. This procedure results into the following optimal control sequence: 
	\begin{subequations}
		\begin{equation} \label{eq:PWAseq}
			U(x)= \begin{cases}
				\mathcal{F}_{1}x+\mathcal{G}_{1}, \mbox{ if }~~ \mathcal{H}_{1}x\leq \mathcal{K}_{1}, \\
				\vdots \\
				\mathcal{F}_{M}x+\mathcal{G}_{M}, \mbox{ if }~~ \mathcal{H}_{M}x\leq \mathcal{K}_{M},
			\end{cases}
		\end{equation}
		where $M$ denotes the number of polyhedral regions over which the sequence is defined, and  
		\begin{align} \label{eq:PWAdetails}
			& \mathcal{F}_{i}=H^{-1}\tilde{G}_{i}'(\tilde{G}_{i}H^{-1}\tilde{G}_{i}')^{-1}\tilde{S}_{i}-H^{-1}F',\\
			& \mathcal{G}_{i} = H^{-1}\tilde{G}_{i}'(\tilde{G}_{i}H^{-1}\tilde{G}_{i}')^{-1}\tilde{W}_{i},\\
			& \mathcal{H}_{i} = 
			\begin{bmatrix}
				& (\tilde{G}_{i}H^{-1}\tilde{G}_{i}')^{-1}\tilde{S}_{i} \\
				& G H^{-1}\tilde{G}_{i}'(\tilde{G}_{i}H^{-1}\tilde{G}_{i}')^{-1}\tilde{S}_{i}-S
			\end{bmatrix}\!, \\
			& \mathcal{K}_{i}=
			\begin{bmatrix}
				& -(\tilde{G}_{i}H^{-1}\tilde{G}_{i}')^{-1}\tilde{W}_{i}\\
				& -GH^{-1}\tilde{G}_{i}'(\tilde{G}_{i}H^{-1}\tilde{G}_{i}')^{-1}\tilde{W}_{i}+W
			\end{bmatrix}\!,
		\end{align}
		with $\tilde{G}_{i},\tilde{W}_{i},\tilde{S}_{i}$ comprising the rows of $G,W,S$ in \eqref{eq:QP2_constr} associated with the $i$-th set of active constraints, for $i=1,\ldots,M$.
	\end{subequations}
	As in the case of implicit MPC, only the first control action is applied to the system, while the other are discarded. Given the input sequence of the form \eqref{eq:PWAseq}, the control action thus results into the \emph{piecewise affine} (PWA) law 
	\begin{equation}\label{eq:explicitMPC}
		u(x)=\begin{cases}
			F_{1}x+g_{1}, \mbox{ if }~~ \mathcal{H}_{1}x \leq \mathcal{K}_{1},\\
			\vdots\\
			F_{M}x+g_{M}, \mbox{ if }~~ \mathcal{H}_{M}x \leq \mathcal{K}_{M},
		\end{cases}
	\end{equation} 
	which can be evaluated at each time instant $t$ by replacing $x$ with the most recent information available on the system state.
	
	\section{Explicit data-driven predictive control}\label{sec:E-DDPC}
	Following the rationale of the previous section, we present here the main results of this paper, leading to a direct translation of the explicit predictive controller in \eqref{eq:explicitMPC} into its \emph{data-driven} counterpart. To this end, we initially recall some results in \cite{DePersis2019} that are instrumental for achieving our goal and, then, exploit these result to obtain a \emph{fully data-driven} explicit predictive controller. Throughout this section, we assume that the noiseless state is measurable, so that the dataset $\mathcal{D}_{T}$ available for the design of the predictive controller comprises a set of noiseless output, \emph{i.e., } $\mathcal{D}_{T}=\{u(t),y(t)\}_{t=0}^{T}$.
	
	\subsection{On the data-based representation of the system} 
	A stepping stone towards the achievement of our goal is given in the following result.
	\begin{lemma}[Fundamental lemma]\label{lemma:1}
		Let the input sequence $\mathcal{U}_{T}$ be persistently exciting of order $n+1$ according to Definition~\ref{def:persistentlyexciting}. Assume that the noiseless dataset $\mathcal{D}_{T}$ is sufficiently long, namely $T \geq (m+1)n+m$. Then, it holds that:
		\begin{equation} \label{eq:persistenlyexciting}
			\mbox{rank} \begin{bmatrix}U_{0,1,T} \\ \hline X_{0,T} \end{bmatrix} = n+m, 
		\end{equation}
		where $
		X_{0,T} = 
		\begin{bmatrix}
			x(0)&\ldots & x(T-1)
		\end{bmatrix}$.
		\hfill$\blacksquare$
	\end{lemma}  
	This result is a direct consequence of \cite[Corollary 2]{Willems2005}, which holds in the noiseless case whenever the design of experiment guiding the data collection phase has been properly performed. 
	
	Assuming that the available data satisfies the assumptions of Lemma~\ref{lemma:1}, the condition in \eqref{eq:persistenlyexciting} is key to the data-driven characterization of the open-loop behavior of system~\eqref{eq:system}. This data-based representation is given in the next theorem, taken from \cite{DePersis2019}.
	\begin{theorem}[Data-driven system representation] \label{th:1}
		Let condition \eqref{eq:persistenlyexciting} hold. Then, \eqref{eq:system} can be equivalently represented as
		\begin{equation}  
			x(t+1) = X_{1,T} \begin{bmatrix}U_{0,1,T} \\ \hline X_{0, T} \end{bmatrix} ^\dagger \begin{bmatrix}u(t)\\ x(t) \end{bmatrix} \label{eq:DD_rep} 
		\end{equation} 
		where $X_{1,T}=\begin{bmatrix}x_{d}(1) & {x_{d}}(2) & \ldots & {x_{d}}(T) \end{bmatrix}$. \hfill$\blacksquare$
	\end{theorem}
	
	\subsection{Towards E-DDPC}
	To retrieve the data-driven explicit controller, we first compute the data-based counterpart of \eqref{eq:QP2}. This can be easily obtained by deriving again the mp-QP from \eqref{eq:MPC} but using the data-driven representation of the system in Theorem \ref{th:1}.

	\begin{lemma}[Data-driven mp-QP] \label{lemma:2}
		If condition~\eqref{eq:persistenlyexciting} holds, the control problem in \eqref{eq:MPC} can be recast as the following data-driven convex mp-QP:
		\begin{subequations} \label{eq:DDQP}
			\begin{align}
				\min_{z} &\quad z'H_{d}z \label{eq:DDQP_cost}\\
				\mbox{s.t.} & \quad G_{d}z\leq W_{d} + S_{d}x, \label{eq:DDQP_constr}
			\end{align}
		\end{subequations}
		with $H_{d},G_{d},W_{d},S_{d}$ depending only upon: the data, the fixed horizons $N_x,N_u$ and $N_c$, the penalties $Q,P$ and $R$ and the given feedback gain.
		\hfill $\blacksquare$
	\end{lemma}
	\begin{proof}
		See Appendix~\ref{appendix1}.
	\end{proof}
	
	Note that the optimization variable in \eqref{eq:DDQP} is given by:
	\begin{equation}\label{eq:DDZexpression1}
		z \triangleq U+H_d^{-1}F_d'x\in\mathbb{R}^{mN_{u}},
	\end{equation}
	with $F_{d}$ being the data-based counterpart of $F$ in \eqref{eq:Zexpression}. We remark that $H_{d}$ is positive definite and, thus, invertible, since the penalties $Q$, $P$ and $R$ have the same features of the ones considered in problem \eqref{eq:MPC}. Therefore, the cost~\eqref{eq:DDQP_cost} is strictly convex and the solution of~\eqref{eq:DDQP} is unique.
	
	By following the same procedure reviewed in Section~\ref{sec:fromItoE}, Lemma~\ref{lemma:2} can be exploited to find a data-driven explicit solution to the predictive control problem, as illustrated next.
	\begin{theorem}[E-DDPC]\label{th:2}
		Let condition~\eqref{eq:persistenlyexciting} hold. Assume that the rows $\tilde{G}_{d}$ of $G_{d}$ in \eqref{eq:DDQP_constr} coupled with active constraints are linearly independent. Then, problem~\eqref{eq:DDQP} admits one and only one solution, resulting in the data-driven explicit predictive law
		\begin{equation}\label{eq:DDexplicit}
			u(x)=\begin{cases}
				F_{d,1}x+g_{d,1}, \mbox{ if } \mathcal{H}_{d,1}x \leq \mathcal{K}_{d,1},\\
				\vdots\\
				F_{d,M}x+g_{d,M}, \mbox{ if } \mathcal{H}_{d,M}x \leq \mathcal{K}_{d,M},
			\end{cases}
		\end{equation}
		which corresponds to the first $m$ components of the control sequence
		\begin{subequations}
			\begin{equation} \label{eq:DDPWAseq}
				U(x)\!=\! \begin{cases}
					\mathcal{F}_{d,1}x+\mathcal{G}_{d,1}, \mbox{ if }~~ \mathcal{H}_{d,1}x\leq \mathcal{K}_{d,1}, \\
					\vdots \\
					\mathcal{F}_{d,M}x+\mathcal{G}_{d,M}, \mbox{ if }~~ \mathcal{H}_{d,M}x\leq \mathcal{K}_{d,M},
				\end{cases}
			\end{equation}
			where   
			\begin{align} 
				& \mathcal{F}_{d,i}=H_{d}^{-1}\tilde{G}_{d,i}'(\tilde{G}_{d,i}H_{d}^{-1}\tilde{G}_{d,i}')^{-1}\tilde{S}_{d,i}-H_{d}^{-1}F_{d}', \label{eq:F}\\
				& \mathcal{G}_{d,i} = H_{d}^{-1}\tilde{G}_{d,i}'(\tilde{G}_{d,i}H^{-1}\tilde{G}_{d,i}')^{-1}\tilde{W}_{d,i},\label{eq:G}\\
				& \mathcal{H}_{d,i}\!=\!\! 
				\begin{bmatrix}
					& \!(\tilde{G}_{d,i}H_{d}^{-1}\tilde{G}_{d,i}')^{-1}\tilde{S}_{d,i} \\
					& \!G_{d} H_{d}^{-1}\tilde{G}_{d,i}'(\tilde{G}_{d,i}H_{d}^{-1}\tilde{G}_{d,i}')^{-1}\tilde{S}_{d,i}-S_{d}
				\end{bmatrix}\!, \\
				& \mathcal{K}_{d,i}\!=\!\!
				\begin{bmatrix}
					& \!-(\tilde{G}_{d,i}H_{d}^{-1}\tilde{G}_{d,i}')^{-1}\tilde{W}_{d,i}\\
					&\! -G_{d}H_{d}^{-1}\tilde{G}_{d,i}'(\tilde{G}_{d,i}H_{d}^{-1}\tilde{G}_{d,i}')^{-1}\tilde{W}_{d,i}\!+\!W_{d}
				\end{bmatrix}\!,
			\end{align}
			with $H_{d}$, $G_{d}$, $W_{d}$ and $S_{d}$ being the data-driven matrices characterizing \eqref{eq:DDQP} and $\tilde{G}_{d,i},\tilde{W}_{d,i},\tilde{S}_{d,i}$ being the rows of $G_{d},W_{d},S_{d}$ associated with the $i$-th set of active constraints, for $i=1,\ldots,M$. \hfill $\blacksquare$
		\end{subequations}
	\end{theorem}
	\begin{proof}
		Since the problem in \eqref{eq:DDQP} is strictly convex, the KKT conditions are necessary and sufficient to characterize optimality. Therefore, to find the solution of \eqref{eq:DDQP} analytically, let us consider the associated KKT conditions, namely:
		\begin{subequations}
			\label{eq:KKT}
			\begin{align}
				& H_dz+G_d'\lambda=0,  \label{eq:KKT1}\\
				& \lambda'(G_dz-W_d-S_dx)=0, \label{eq:KKT2}\\
				& \lambda\geq 0 ,\label{eq:KKT3} \\
				& G_dz\leq W_d+S_dx, \label{eq:KKT4}
			\end{align}
		\end{subequations}
		where $\lambda$ is the vector of Lagrange multipliers associated with the inequality constraint in \eqref{eq:DDQP_constr}. 
		From the stationarity condition in \eqref{eq:KKT1}, we can derive the following relationship between $z$ in \eqref{eq:Zexpression} and $\lambda$:
		\begin{equation} \label{eq:z}
			z=-H_d^{-1}G_d' \lambda,
		\end{equation}
		that, in turn, allows us to recast the complementary slackness condition in \eqref{eq:KKT2} as
		\begin{equation*}
			\lambda'(-G_dH_d^{-1}G_d'\lambda-W_d+S_dx)=0.
		\end{equation*} 
		Let $\bar{\lambda}$ be the subset of Lagrange multipliers coupled with the inactive constraints and $\tilde{\lambda}$ the remaining active ones. By combining complementary slackness (see \eqref{eq:KKT2}) and the dual feasibility condition \eqref{eq:KKT3}, the Lagrange multipliers $\bar{\lambda}$ turn out to be zero. Moreover, straightforward manipulations of the above equality allow us to equivalently define $\tilde{\lambda}$ as:
		\begin{equation}
			\tilde{\lambda}=-(\tilde{G}_dH_d^{-1}\tilde{G}_d')^{-1}(\tilde{W}_d+\tilde{S}_dx),
		\end{equation} 
		where $\tilde{G}_{d}$, $\tilde{W}_{d}$ and $\tilde{S}_{d}$ collect the rows of $G_{d}$, $W_{d}$ and $S_{d}$ associated with active constraints, respectively. Since the rows of $\tilde{G}_{d}$ are assumed to be linearly independent, it holds that
		\begin{equation} \label{eq:DDzexpression}
			z=H_d^{-1}\tilde{G}_d'(\tilde{G}_dH_d^{-1}\tilde{G}_d')^{-1}(\tilde{W}_d+\tilde{S}_dx),
		\end{equation}
		from which straightforward manipulations result into \eqref{eq:F}-\eqref{eq:G}. The primal and dual feasibility conditions in \eqref{eq:KKT3} and \eqref{eq:KKT4} allow us to explicitly find the regions of the state space where \eqref{eq:DDzexpression} holds, which are defined as
		\begin{subequations}
			\label{eq:regioncharacterization}
			\begin{align}
				& -(\tilde{G}_dH_d^{-1}\tilde{G}_d')^{-1}(\tilde{W}_d+\tilde{S}_dx)\geq 0 \label{eq:positivity}\\
				& G_d  H_d^{-1}\tilde{G}_d'(\tilde{G}_dH_d^{-1}\tilde{G}_d')^{\!-1}(\tilde{W}_d\!+\!\tilde{S}_dx) \!\leq\! W_d\!+\!S_dx \label{eq:constrainthold}.
			\end{align}
		\end{subequations}
		By relying on \eqref{eq:DDzexpression} and \eqref{eq:regioncharacterization} and considering all possible combinations of active constraints, straightforward manipulations result into the explicit control sequence in \eqref{eq:DDPWAseq}, thus concluding the proof. 
	\end{proof}
	
	\begin{remark}
		In case of \emph{degeneracy}, namely when the combinations of active constraints lead to linearly dependent rows in $\tilde{G}_{d}$, the problem can still be handled by exploiting an approach similar to the one in \cite{Bemporad2002b}. Even in this scenario, no identification step would be required to explicitly solve the predictive control problem. \hfill $\blacksquare$ 
	\end{remark}
	
	Lemma~\ref{lemma:2} and Theorem~\ref{th:2} allows us to further infer the following results.
	\begin{theorem}[Model/Data equivalence]\label{th:3}
		Let the assumptions of Theorem~\ref{th:2} hold. Then the data-driven explicit law in \eqref{eq:DDexplicit} is equivalent to the model-based predictive controller in \eqref{eq:explicitMPC}. \hfill $\blacksquare$
	\end{theorem}
	
	\begin{proof}
		Based on the results of Lemma~\ref{lemma:2}, the data-driven predictive control problem \eqref{eq:DDQP} originating the explicit law in \eqref{eq:DDexplicit} is equivalent to \eqref{eq:QP2} and, thus, to the MPC problem in \eqref{eq:MPC}. Since the steps leading to the explicit controller in \eqref{eq:DDexplicit} are the same performed to obtain the model-based explicit predictive controller, the equivalence straightforwardly follows.
	\end{proof} 
	%
	\begin{lemma}[Continuity]\label{lemma:3}
		The data-driven PWA control law in \eqref{eq:DDexplicit} is continuous over the boundaries of the polyhedral regions characterizing it. \hfill $\blacksquare$ 
	\end{lemma}
	\begin{proof}
		The continuity over the boundaries of the polyhedral regions can be inferred from the properties of the model-based explicit predictive controller (see \cite{Bemporad2002b}) and its equivalence with the data-driven solution, as dictated by Theorem~\ref{th:3}. A formal proof is thus omitted, as it straightforwardly follows from the above results.
	\end{proof}
	
	\section{Practical implementation and noise handling}\label{sec:practice}
	The main steps to compute the E-DDPC law are summarized in Algorithm~\ref{algo1}. After an initial phase in which the data are manipulated to cast the data-driven mp-QP problem (see steps~\ref{step:1}-\ref{step:4}), one has to check if the considered problem is characterized by degenerate situations, which can be handled by exploiting the same procedure proposed in \cite{Bemporad2002b} without requiring any prior identification of a model for the system $\mathcal{S}$ (see step~\ref{step:6.1}). At step~\ref{step:7}, Theorem~\ref{th:2} can then be directly applied to retrieve the explicit control law. Since the result in Theorem~\ref{th:2} relies on the explicitation of the KKT conditions for problem~\eqref{eq:DDQP}, it leads to an enumeration of all possible combinations of active constraints. In turn, this procedure might result in an overly-complex PWA controller. To overcome this limitation, at step~\ref{step:8} polyhedral regions characterized by the same control law are merged by following the approach in \cite{Bemporad2001}. 
	
	Once Algorithm~\ref{algo1} has been run, the control action at time $t$ solely requires to $(i)$ explore all polyhedral regions characterizing the reduced law \eqref{eq:DDexplicit} in order to locate the one the current state $x(t)$ belongs to, and $(ii)$ compute the corresponding state-feedback affine input. The computational time required for this operation increases with the horizons $N_{c}$, $N_{u}$ and with the number of inputs and states, since the latter is assumed to be fully measurable. Therefore, also the data-driven version of the explicit solution is mainly appealing when short horizons or blocking control moves are used~\cite{Alessio2009}. 
	\begin{algorithm}[!tb]
		\caption{Noiseless E-DDPC: offline procedure}
		\label{algo1}
		~\textbf{Input}: Dataset $\mathcal{D}_{T}$; penalties $Q, P \succeq 0$; $R \succ 0$; horizons $N_{x}, N_{u}, N_{c}\!>\!0$; constraints $\mathcal{C}_{x},\mathcal{C}_{u}$; feedback gain~$K$.
		\vspace*{.1cm}\hrule\vspace*{.1cm}
		\begin{enumerate}[label=\arabic*., ref=\theenumi{}]
			\item\label{step:1} \textbf{Construct} the data-based matrix in \eqref{eq:persistenlyexciting}.
			\item\label{step:2} \textbf{Build} $H_{d}$, $G_{d}$, $W_{d}$, $S_{d}$ in \eqref{eq:DDQP} based on the chosen cost and constraints.
			\item \textbf{Find} all possible combinations of active constraints. 
			\item\label{step:4} \textbf{Isolate} the matrices $\tilde{G}_{d}$, $\tilde{W}_{d}$ and $\tilde{S}_{d}$ comprising the rows of $G_{d}$, $W_{d}$, $S_{d}$ associated to the sets of active constraints
			\item \textbf{If not} all rows of $\tilde{G}_{d}$ are \textbf{linearly independent}
			\begin{enumerate}[label=\theenumi{}.\arabic*., ref=\theenumi{}.\theenumii{}]
				\item\label{step:6.1} \textbf{Handle} the degeneracy, \emph{e.g.,} as in \cite{Bemporad2002b}.
			\end{enumerate}			
			\item\label{step:7} \textbf{Find} the PWA explicit law as in Theorem~\ref{th:2}.
			\item\label{step:8} \textbf{Merge} polyhedral regions whenever possible, \emph{e.g.,} with the approach proposed in \cite{Bemporad2001}.
		\end{enumerate}
		\vspace*{.1cm}\hrule\vspace*{.1cm}
		~\textbf{Output}: Optimal explicit law $u(x)$.
	\end{algorithm}
	
	\subsection{Stability and recursive feasibility}
	Among the penalties characterizing the predictive cost in \eqref{eq:MPCcost}, it is known that the choice of the terminal weight $P \succeq 0$ and the static feedback $K$, dictating the input for $k\geq N_{u}$, influence the stability properties of the predictive controller in \eqref{eq:DDexplicit} \cite{Bemporad2002b}. Based on the existing guidelines for their choice in the model-based case, we can obtain their data-driven counterparts as follows. 
	
	If the system is known to be open-loop stable, it is possible to select $K=0$ and set $P$ as the solution to the data-driven Lyapunov equation
	\begin{equation}\label{eq:Lyapunov1}
		P=\xi_d'P\xi_d+Q,
	\end{equation}
	with 
	\begin{equation}\label{eq:xid}
		\xi_d = 
		X_{1,T} \begin{bmatrix}U_{0,1,T} \\ \hline X_{0, T} \end{bmatrix} ^\dagger \begin{bmatrix}0_{m\times n} \\ I_{n} \end{bmatrix}.
	\end{equation}
	When the system is not known to be open-loop stable or it is known to be unstable, the terminal penalty $P$ and the feedback gain $K$ can be instead selected as the solutions of the linear quadratic regulation (LQR) problem, as discussed in \cite{DePersis2021}. When these choices are performed, the following asymptotic results can be directly inferred from the model-based ones.
	\begin{theorem}[Stability and feasibility]\label{lemma:4}
		Let condition \eqref{eq:persistenlyexciting} hold. Let $N_x=\infty$, $K=0$ or $K$ be the LQR gain obtained as in \cite{DePersis2021}, and $N_{c}<\infty$ be sufficiently large to guarantee the existence of feasible input sequences at each time step. Then, the predictive control law resulting from the solution of \eqref{eq:DDQP} asymptotically stabilizes the system in \eqref{eq:system}, while enforcing the fulfillment of constraints from all initial states $x$ such that the optimization problem is feasible at time $t=0$. \hfill $\blacksquare$
	\end{theorem}
	\begin{proof}
		This result stems from the fact that Lemma~\ref{lemma:2} guarantees that the data-driven problem in \eqref{eq:DDQP} is exactly equal to the model-based one in \eqref{eq:QP2}. The latter corresponds to the one in \eqref{eq:MPC}, for which a similar asymptotic result hold as shown in \cite{Bemporad2002b}. The proof easily follows from this concatenation of equalities.
	\end{proof}
	\subsection{Tracking E-DPCC}\label{sec:tracking}
	As for the model-based case, problem~\eqref{eq:DDQP} can be extended to attain offset-free tracking of a user-defined reference signal. This entails a change in the cost to be optimized from the one shown in \eqref{eq:MPCcost} to
	\begin{subequations}\label{eq:trackingDDPC}
		\begin{equation}
			\sum_{k=0}^{N_{x}-1} \left[\|x(k)-r(t)\|_{Q}^{2}+\|\delta u(k)\|_{R}^{2}\right],
		\end{equation}
		where $\delta u(k)$ is defined as
		\begin{equation}
			u(k)=u(k-1)+\delta u(k), ~~ k \geq 0,
		\end{equation}
		and it can be eventually subject to polytopic constraints for $0 \leq k \leq N_{u}$, while it satisfies the following
		\begin{equation}
			\delta u(k)=0,~~\forall k \geq N_{u}.
		\end{equation}
	\end{subequations}
	This reformulation leads to a data-driven problem similar to the one in \eqref{eq:DDQP}, with a data-driven input increment $\delta u(x)$ that depends on the extended vector
	\begin{equation}\label{eq:augm_vector}
		\begin{bmatrix}
			x(t)' & u(t-1)' & r(t)'
		\end{bmatrix}'.
	\end{equation}
	Such a vector has to be considered instead of the state $x(t)$ to find the optimal control action at time $t$. We remark that this formulation entails that no preview of the reference is available, so that the set point is frozen to $r(t) \in \mathbb{R}^{n}$ over the state horizon $N_{x}$. If the set point to be tracked is known in advance, the cost can be modified by replacing $r(t)$ with $r(t+k)$ and augmenting the extended vector in \eqref{eq:augm_vector} accordingly. 
	\subsection{Handling noise in E-DDPC}\label{sec:handlenoise}
	All results shown in the previous sections are derived in the ideal case of noiseless data. Nonetheless, based on the properties of the noisy dataset $\mathcal{D}_{T}^{n}$ introduced in Section~\ref{sec:problem}, our findings can be extended to the noisy case by relying on the following lemma.
	\begin{lemma}[Consistency]\label{lemma:5}
		Assume $L$ noisy dataset $D_{T}^{n,l}=\{\mathcal{U}_{T},\mathcal{Y}_{T}^{n,l}\}$, with $l=1,\ldots,L$, can be gathered by exciting the system with the same input sequence, while observing different realization of the measurement noise. Assume that the measurement noise is white and with zero mean. Then, given the definition of the noisy outputs in \eqref{eq:noisy_y}, the following asymptotic result holds:
		\begin{equation}\label{eq:lln}
			\lim_{L \rightarrow \infty} \frac{1}{L} \sum_{\ell=1}^{L} y^{n}(t;l)=y(t), ~~ \forall t=0,\ldots,T.
		\end{equation}
		\hfill $\blacksquare$
	\end{lemma}
	\begin{proof}
		Based on \eqref{eq:noisy_y}, each realization $l$ of the output at time $t$ corresponds to
		\begin{equation*}
			y^{n}(t;l)=y(t)+v(t;l),
		\end{equation*}
		where $y(t)$ is the noiseless output and $v(t;l)$ is the $l$-th realization of the measurement noise, for $l=1,\ldots,L$. By replacing this definition into the samples mean on the left-hand-side of \eqref{eq:lln}, we obtain
		\begin{equation*}
			\frac{1}{L} \sum_{\ell=1}^{L} y^{n}(t;l)=\frac{1}{L} \sum_{\ell=1}^{L} (y(t)+v(t;l))=y(t)+\frac{1}{L}\sum_{l=1}^{L} v(t;l).
		\end{equation*}
		Because of the assumptions on the measurement noise, from the law of large numbers it straightforwardly follows that
		\begin{equation*}
			\lim_{L \rightarrow \infty} \frac{1}{L} v(t;l)=0,
		\end{equation*} 
		resulting in the asymptotic result in \eqref{eq:lln}.
	\end{proof}
	
	Under the assumption that $L$ experiments can be performed on the system by applying to it the same persistently exciting input sequence, we thus construct the \emph{averaged} dataset $\bar{\mathcal{D}}_{T}^{L}=\{\mathcal{U}_{T},\bar{\mathcal{Y}}_{T}^{L}\}$, with $\bar{\mathcal{Y}}_{T}^{L}=\{\bar{y}^{L}(t)\}_{t=0}^{T}$ comprising the average outputs
	\begin{equation}\label{eq:avg_out}
		\bar{y}^{L}(t)=\frac{1}{L}\sum_{l=1}^{L} y^{n}(t;l),~~\forall t=0,\ldots,T.
	\end{equation} 
	This dataset is at the core of the following asymptotic equivalent result.
	\begin{theorem}[Model/Noisy data equivalence]\label{th:4}
		Let the assumptions of Lemma~\ref{lemma:5} hold. Consider the averaged set $\bar{\mathcal{D}}_{T}^{L}=\{\mathcal{U}_{T},\bar{\mathcal{Y}}_{T}^{L}\}$ and let $\bar{X}_{0,T}$ and $\bar{X}_{1,T}$ be the following collection of averaged states
		\begin{subequations}
			\label{eq:avg_datamatrices}
			\begin{align}
				& \bar{X}_{0,T}=[\bar{x}(1),...,\bar{x}(T-1)],\\
				& \bar{X}_{1,T}=[\bar{x}(2),...,\bar{x}(T)].
			\end{align}
		\end{subequations}
		For $L \rightarrow \infty$, the noisy predictive controller obtained by explicitly solving \eqref{eq:DDQP} with $\bar{X}_{0,T}$ and $\bar{X}_{1,T}$ respectively replacing $X_{0,T}$ and $X_{1,T}$ in \eqref{eq:DD_rep} converges to the noiseless solution in \eqref{eq:DDexplicit}. \hfill $\blacksquare$
	\end{theorem}
	\begin{proof}
		According to Lemma~\ref{lemma:5}, the averaged dataset and the noiseless one asymptotically coincide. Therefore, for $L\rightarrow \infty$, the data-based matrices used to construct the noisy data-driven explicit controller are equal to the noiseless one, from which the convergence straightforwardly follows.
	\end{proof}
	
	As summarized in Algorithm~\ref{algo2}, we thus propose to perform the same experiment multiple times, average the available data and then run Algorithm~\ref{algo1} by using the averaged dataset. We remark that the larger $L$ is, the more likely the asymptotic result is to hold. It is thus crucial to perform as many experiments as possible, up to the bound typically dictated by practical limitations. Note that, by repeatedly performing the same experiment and then averaging out the measured outputs, we preserve the characteristics of the input sequence, thus guaranteeing that condition \eqref{eq:persistenlyexciting} is still verified.   
	
	\begin{algorithm}[!tb]
		\caption{Noisy E-DDPC: offline steps}
		\label{algo2}
		~\textbf{Input}: Sequence $\mathcal{U}_{T}$; number of experiments $L\!\geq\!1$; penalties $Q, P \succeq 0$; $R \succ 0$; horizons $N_{x}, N_{u}, N_{c}\!>\!0$; constraints $\mathcal{C}_{x},\mathcal{C}_{u}$; feedback gain~$K$.
		\vspace*{.1cm}\hrule\vspace*{.1cm}
		\begin{enumerate}[label=\arabic*., ref=\theenumi{}]
			\item \textbf{Perform} $L$ experiments, by feeding $\mathcal{S}$ with $\mathcal{U}_{T}$.
			\item \textbf{Store} the outputs $\{y^{n}(t;l)\}_{t=0}^{T}$, for $l=1,\ldots,L$.
			\item \textbf{Build} the averaged dataset $\bar{\mathcal{D}}_{T}^{L}$ according to \eqref{eq:avg_out}.
			\item \textbf{Run} Algorithm~\ref{algo1} by exploiting $\bar{\mathcal{D}}_{T}^{L}$.
		\end{enumerate}
		\vspace*{.1cm}\hrule\vspace*{.1cm}
		~\textbf{Output}: Noisy explicit law $u^{n}(x)$.
	\end{algorithm}
	
	\section{Simulation examples}\label{sec:examples}
	The performance of E-DDPC are now assessed on three benchmark simulation examples: two numerical case studies of regulation to zero, i.e., the open-loop stable system of \cite{Bemporad2002b} and the sparse unstable system of \cite{Dean2020,Breschi2021}, and a more realistic tracking application, namely the altitude control of a quadcopter. 
	In the last two examples, the data are collected in closed-loop, assuming the systems to be preliminarily stabilized by an existing (unknown) controller. For simplicity, we impose $N_{x}=N_{u}=N_{c}=N$. All computations are carried out on an Intel Core i7-7700HQ processor, running MATLAB 2019b. 
	\subsection{Open-loop stable system}
	\begin{table*}[!tb]
		\caption{Open-loop stable system example: RMSE$_{\mathcal{O}}$ in \eqref{eq:rmse} (mean $\pm$ standard deviation) \emph{vs} $L$.}\label{tab:OL}
		\centering
		\begin{tabular}{cccccc}
			$L$ & 1 & 5 & 10 & 50 & 100\\
			\hline
			RMSE$_{\mathcal{O}}$ & 0.075 $\pm$ 0.171 & 0.022 $\pm$ 0.016 & 0.020 $\pm$ 0.015 & 0.008 $\pm$ 0.008 & 0.006 $\pm$ 0.005\\
			\hline 
		\end{tabular}
	\end{table*}
	\begin{figure}[!tb]
		\centering
		\begin{tabular}{c}
			\subfigure[State trajectories]{\includegraphics[scale=0.8,trim=2cm .8cm 4cm 1cm,clip]{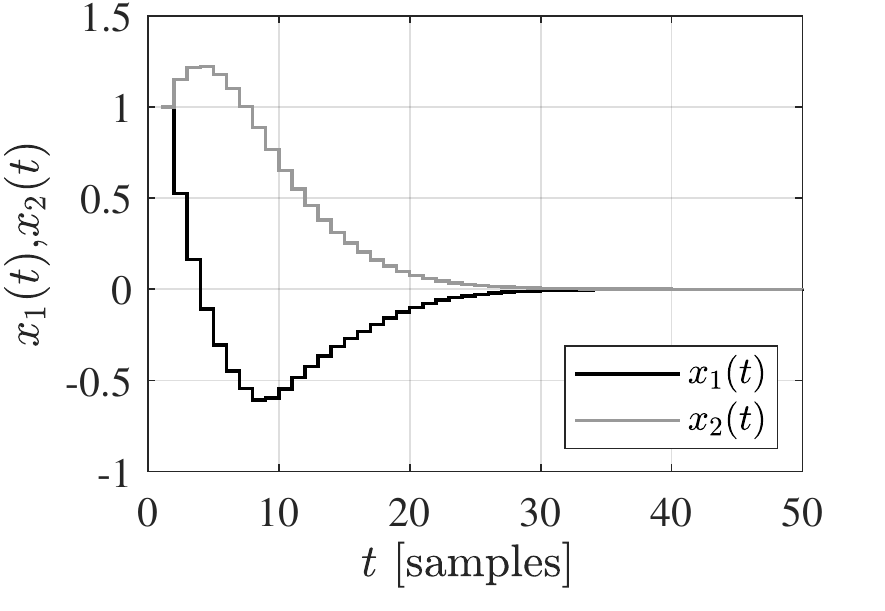}}\\ \subfigure[Input]{\includegraphics[scale=0.8,trim=2cm .8cm 4cm 1cm,clip]{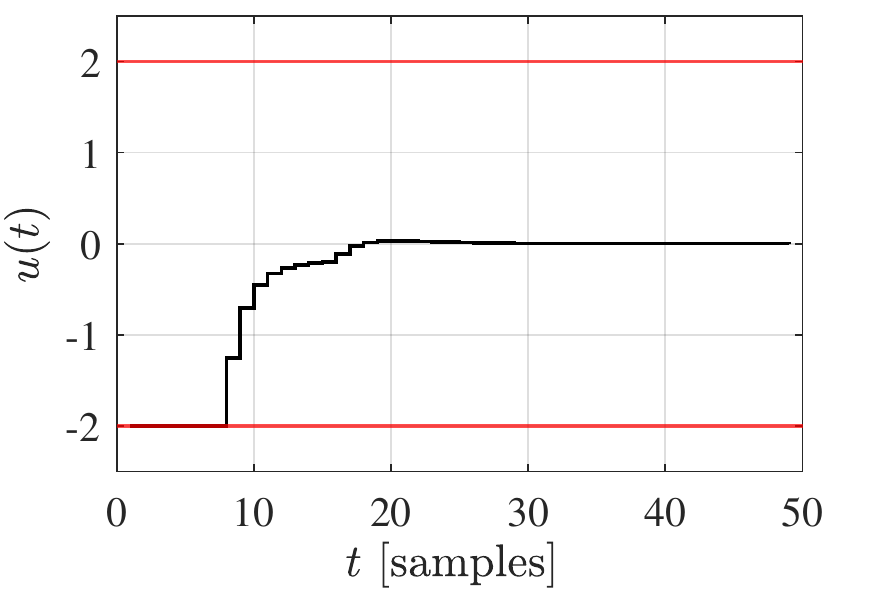}}
		\end{tabular}
		\caption{Open-loop stable system example: state and input trajectories obtained with E-DDPC.}\label{fig:OL}
	\end{figure}
	\begin{figure}[!tb]
		\includegraphics[scale=.5,trim=2cm 9.5cm 2cm 9.5cm,clip]{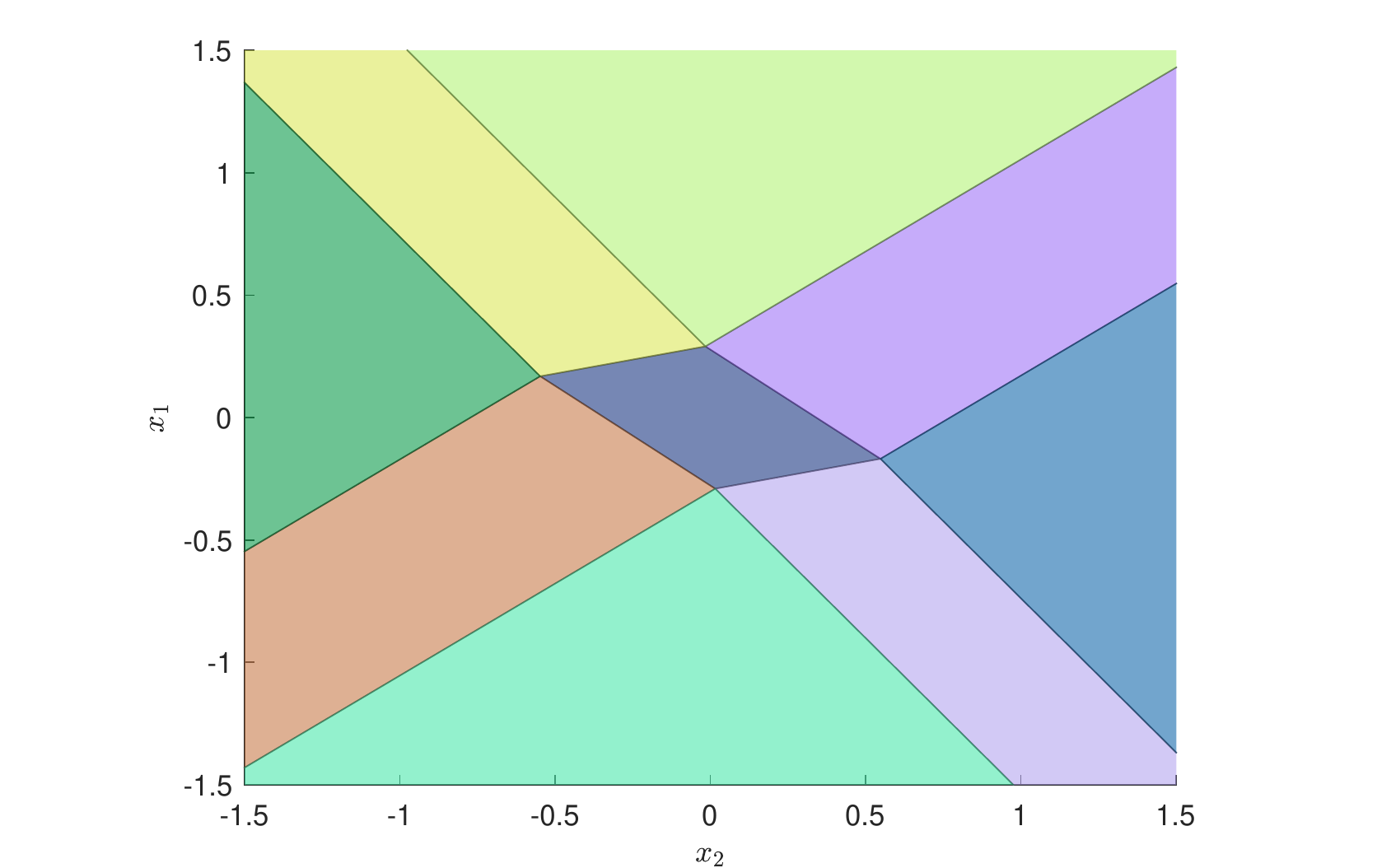}
		\caption{Open-loop stable system example: polyhedral partition of the explicit data-driven law, plotted with the Hybrid Toolbox~\cite{HybTBX}.\label{fig:OL_partition}}
	\end{figure}
	Consider the system introduced in \cite{Bemporad2002b}, the dynamics of which is characterized by the following system of difference equations:
	\begin{equation}
		x(t+1)=\begin{bmatrix}
			0.7326 & -0.0861\\
			0.1722 & 0.9909
		\end{bmatrix}x(t)+\begin{bmatrix}
			0.0609\\
			0.0064
		\end{bmatrix}u(t),
	\end{equation}
	where we assume the state to be fully measurable. $L$ experiments are carried out by always exciting the system with a random input sequence of length $T=20$, uniformly distributed within the interval $[-5,5]$, while the output is assumed to be corrupted by a zero-mean white noise sequence with standard deviation 0.024\footnote{This corresponds to a Signal-to-Noise Ratio (SNR) around $20$ dB, averaged with respect to the $L$ dataset and the two output components.}. As in \cite{Bemporad2002b}, our task is to steer the system's state to the origin, while satisfying the following input constraint
	\begin{equation*}
		-2 \leq u(k) \leq 2.
	\end{equation*}
	The cost of the optimal control problem is characterized by the penalties $Q=I_{2}$ and $R=0.01$, while the terminal weight $P$ is selected by solving the data-driven Lyapunov equation in \eqref{eq:Lyapunov1}. Nonetheless, differently from \cite{Bemporad2002b}, we assume the state, control and constraint horizon to be equal, setting all of them to $N=2$. 
	
	E-DDPC is designed for increasing values of $L$ by carrying out a Monte Carlo analysis with $20$ different realizations of the input used to construct the $L$ datasets and the corresponding measurement noise. This allows us to assess the robustness of the approach to different realizations of the persistently exciting input fed to the system and the effectiveness of the proposed noise management strategy. The results of this Monte Carlo analysis are shown in \tablename{~\ref{tab:OL}}, where the quality of the attained closed-loop performance are assessed over a noiseless test by looking at the following indicator: 
	\begin{equation}
		\mbox{RMSE}_{\mathcal{O}}=\frac{1}{n}\sum_{i=1}^{n}\sqrt{\frac{1}{T_{v}} \sum_{t=0}^{T_{v}-1}(x_{i}(t)-x_{i}^{\star}(t))^{2}}, \label{eq:rmse}
	\end{equation}
	with $n=2$, which allows us to compare the obtained state trajectory with the ideal one $\{x^{\star}(t)\}_{t=0}^{T_{v}-1}$, retrieved by using the \emph{oracle} explicit MPC law $\mathcal{O}$, i.e., using the real model of the system as in \cite{Bemporad2002b}. Clearly, augmenting the number of experiments $L$ used to construct the averaged dataset leads to state trajectory that increasingly matches (on average) the one resulting from the application of the oracle explicit MPC, with a corresponding reduction in the standard deviation of the results obtained over the $20$ realizations of the datasets. The state trajectories, input and partition associated to E-DDPC obtained for $L=50$ are respectively shown in \figurename{s~\ref{fig:OL}-\ref{fig:OL_partition}}. Both the behavior of the system and the obtained partition are almost identical to the ones shown in \cite{Bemporad2002b}, as expected from the theoretical results of Section~\ref{sec:E-DDPC}, with the partition characterizing the explicit law being characterized by $9$ polyhedral regions.   
	
	\subsection{Sparse unstable system}
	\begin{table}[!tb]
		\caption{Sparse unstable system example: RMSE$_{\mathcal{O}}$ \eqref{eq:rmse} for increasing noise levels.} \label{tab:1}
		\centering
		\begin{tabular}{cccccc}
			$\overline{\mbox{SNR}}$ [dB] \hspace*{-.2cm}&\hspace*{-.2cm} 40 \hspace*{-.2cm}&\hspace*{-.2cm} 30 \hspace*{-.2cm}&\hspace*{-.2cm} 19.9 \hspace*{-.2cm}&\hspace*{-.2cm} 10 \hspace*{-.2cm}&\hspace*{-.2cm} 4.6\\ 
			\hline
			RMSE$_{\mathcal{O}}$ \hspace*{-.2cm}&\hspace*{-.2cm} 6.4$\cdot 10^{-5}$ \hspace*{-.2cm}&\hspace*{-.2cm} 3.1$\cdot 10^{-4}$ \hspace*{-.2cm}&\hspace*{-.2cm} 1.1$\cdot 10^{-3}$ \hspace*{-.2cm}&\hspace*{-.2cm} 4.9$\cdot 10^{-3}$ \hspace*{-.2cm}&\hspace*{-.2cm} 1.9$\cdot 10^{-2}$ \\ 
			\hline
		\end{tabular}
	\end{table}
	\begin{figure}[!tb]
		\centering
		\includegraphics[scale=0.8,trim=1.25cm 1cm 0cm 1cm,clip]{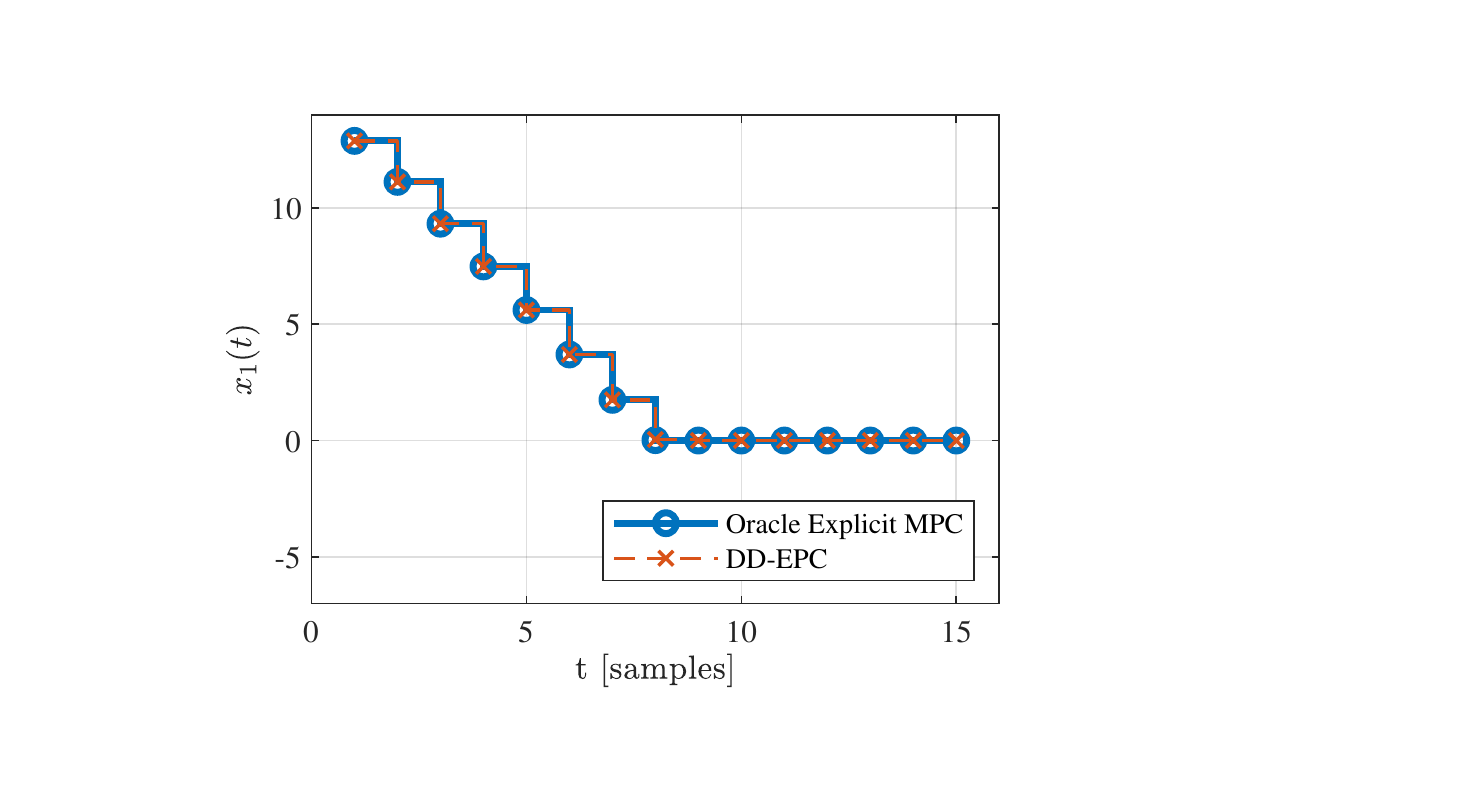}
		\caption{Sparse unstable example: evolution of the first component of the state, E-DDPC \emph{vs} Oracle explicit MPC $\mathcal{O}$. The two trajectory are almost always overlapped.}\label{fig:ex1_statecomp}
	\end{figure}
	Let the linear \emph{multi-input multi-output} (MIMO) data-generating system be characterized by the state-space equations in \eqref{eq:system}, with
	\begin{equation}
		A=
		\begin{bmatrix}
			1.01 & 0.01 & 0\\
			0.01& 1.01 & 0.01\\
			0 & 0.01 & 1.01
		\end{bmatrix},~~
		B=I_{n}.
	\end{equation}
	The explicit predictive law for this system is designed by setting $N=3$, $R=0.01 I_{3}$ and $Q=P=I_{3}$, while imposing only the following box constraints on the inputs:
	\begin{equation}
		-2 \leq u_{i}(k) \leq 2,~~k=0,1,2,~~i=1,2,3.
	\end{equation}
	To retrieve the E-DDPC, we collect $L=10$ datasets of length $T=200$ by stabilizing the system with the static law introduced in \cite{Breschi2021}, namely
	\begin{equation*}
		u(t)=-I_{3}x(t)+I_{3}r(t),
	\end{equation*} 
	selecting $r(t)$ uniformly at random within the interval $[-5,10]$, so as to guarantee that the input fed to the plant is persistently exciting according to Definition~\ref{def:persistentlyexciting}. The measured output is then corrupted by zero mean white noise with variance $\Sigma$, whose effect on the data is evaluated through the average Signal-to-Noise Ratio (SNR) over the three output channels and the $L$ datasets, \emph{i.e.,}
	\begin{equation*}
		\overline{\mbox{SNR}}\!=\!\frac{1}{3L} \sum_{i=1}^{3} \sum_{l=1}^{L} 10\log{\frac{\sum_{t=0}^{T}(x_{i}(t;l)\!-\!v_{i}(t;l))^{2}}{\sum_{t=0}^{T} v_{i}(t;l)^{2}}},~\mbox{[dB]}.
	\end{equation*}
	The available $2000$ samples are used to construct the averaged dataset according to \eqref{eq:avg_out}, so as to handle the noise via the strategy proposed in Section~\ref{sec:handlenoise}. By focusing on a test of length $T_{v}=15$ samples, we assess the performance of E-DDPC with the noise management approach for increasing levels of noise. This evaluation is performed by considering two performance indicators, namely the one in \eqref{eq:rmse} and 
	\begin{align}
		& \mbox{RMSE}_{0}=\frac{1}{3}\sum_{i=1}^{3}\sqrt{\frac{1}{T_{v}} \sum_{t=0}^{T_{v}-1} x_{i}(t)^{2}},
	\end{align}  
	with the latter allowing us to assess (on average) the capability of E-DDPC to bring the states of the system to the origin. As shown in \tablename{~\ref{tab:1}}, the proposed noise handling strategy allows us reproduce quite tightly the trajectory resulting from using the explicit MPC computed as in \cite{Bemporad2002b}, with the result becoming relatively sensitive to noise only for an average SNR lower than $10$~dB. Instead, the degradation on the regulation performance due to noise is almost negligible, since RMSE$_{0} \approx 5.5$ independently from the features of the noise affecting the set used to design E-DDPC. Note that the latter result is mainly due to the randomly chosen initial state considered in the tests\footnote{$x(0)=\begin{bmatrix}
			12.88 & 10.95 & -14.44
		\end{bmatrix}'$.} and, thus, to the initial transient of the state trajectories. These conclusions are further supported by the comparison shown in \figurename{~\ref{fig:ex1_statecomp}}, obtained when the training set is corrupted by noise yielding an average SNR of $20$~dB\footnote{In this case, the resulting partition comprises $791$ regions.}. As expected, the differences between the state trajectories resulting from the use of the data-driven explicit predictive controller and the model-based ones turn out to be negligible\footnote{For the sake of visualization, we solely show the trajectory of the first state, since it reflect the behaviors of the remaining states.}, thus confirming the effectiveness of E-DDPC and the proposed noise handling strategy.  
	
	\subsection{Altitude control}
	\begin{table}
		\caption{Altitude control example: parameters of the quadcopter and their physical meaning.\label{tab:2}}
		\centering
		\begin{tabular}{ccc}
			\textbf{Name \& symbol} & \textbf{Value} & \textbf{m.u.}\\
			\hline
			Mass, $m$ & 0.5 & kg \\
			\hline
			Inertia on x, $I_{x}$ & $5 \cdot 10^{-3}$ & Nms$^2$\\
			\hline
			Inertia on y, $I_{y}$ & $5 \cdot 10^{-3}$ & Nms$^2$\\
			\hline 
			Inertia on z, $I_{z}$ & $9 \cdot 10^{-3}$ & Nms$^2$\\
			\hline
			Motor inertia, $J_{m}$ & $3.4 \cdot 10^{-5}$ & Nms$^2$\\
			\hline
			Drag factor, $d$ & $1.1 \cdot 10^{-5}$ & Nms$^2$\\
			\hline
			Thrust factor, $b$ & $7.2 \cdot 10^{-5}$ & Ns$^2$\\
			\hline
			Center-to-propeller distance, $l$ & $0.25$ & m\\
			\hline
			Gravitational acceleration, $g$ & 9.81 & m/s$^2$\\
			\hline
		\end{tabular}
	\end{table}
	\begin{table}
		\caption{Altitude control example: parameters of the controllers in \eqref{eq:PIDcontrolquadrotor}.}\label{tab:3}
		\centering
		\begin{tabular}{cccccccc}
			$k_{1}$ & $k_{2}$ & $k_{3}$ & $k_{4}$ & $k_{5}$ & $k_{6}$ & $k_{7}$ & $k_{8}$\\
			\hline
			4 & 2 & 4 & 2 & 4 & 2 & 2 & 2\\
			\hline
		\end{tabular}
	\end{table}
	\begin{figure}[!tb]
		\centering
		\begin{tabular}{c}
			\subfigure[Takeoff]{\includegraphics[scale=.8,trim=1.85cm .8cm 4cm 1cm,clip]{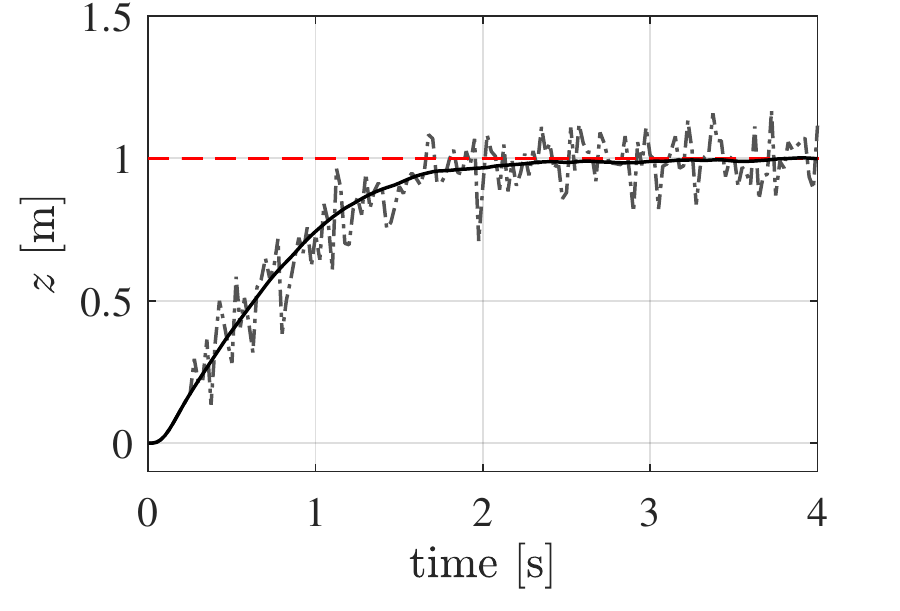}}\\
			\subfigure[Landing]{\includegraphics[scale=.8,trim=1.85cm .8cm 4cm 1cm,clip]{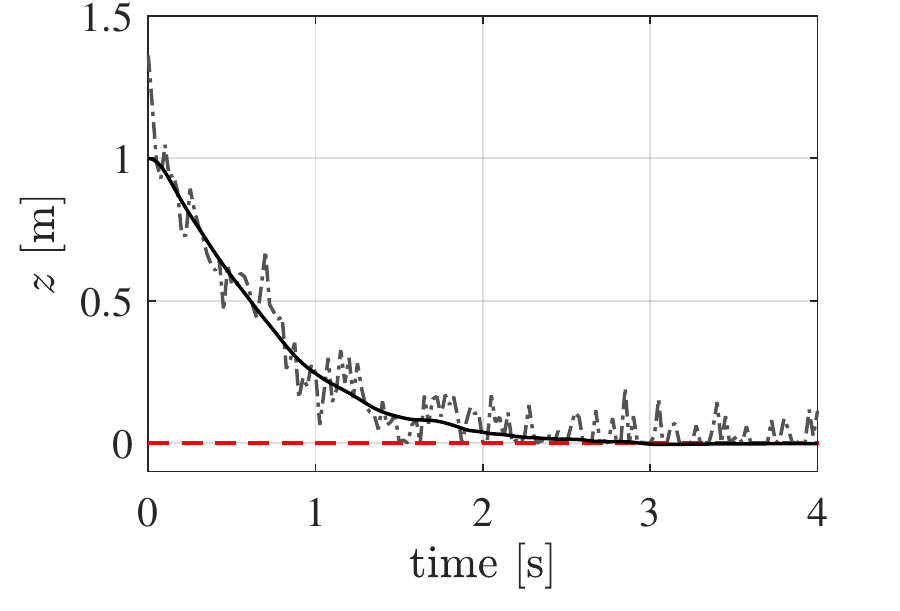}}
		\end{tabular}
		\caption{Altitude control example: measured (dotted-dashed gray line) and actual (black line) altitude \emph{vs} reference (dashed red line).}\label{Fig:noisy_testQUAD}
	\end{figure}
	We now consider the problem of controlling the altitude of a quadcopter, by considering the same dynamical model proposed in \cite{Formentin2011} to perform the data-collection experiments and to assess the effectiveness of the altitude E-DDPC. The data-generating system is thus described by the following equations
	\begin{subequations} \label{eq:quadrotormodel}
		\begin{equation}
			\begin{aligned}
				\ddot{x} &= U_{1}{\frac{(\cos\psi\cos\phi\sin\theta+\sin\psi\sin\phi)}{m}}, \\
				\ddot{y} &= U_{1}{\frac{(\sin\psi\cos\phi\sin\theta-\sin\phi\cos\psi)}{m}},\\
				\ddot{z} &= U_{1}{\frac{(\cos\theta\cos\phi)}{m}}-g,  
			\end{aligned}
		\end{equation}
		\begin{equation}
			\begin{aligned}
				\dot{p} &= {\frac{(I_{y}-I_{z})}{I_{x}}}qr+{\frac{1}{I_{x}}}U_{2}-{\frac{J_m}{I_{x}}}q~\Omega_{R}, \\
				\dot{q} &= {\frac{(I_{z}-I_{x})}{I_{y}}}pr+{\frac{1}{I_{y}}}U_{3}+{\frac{J_m}{I_{y}}}p~\Omega_{R}, \\
				\dot{r} &= {\frac{(I_{x}-I_{y})}{I_{z}}}pq+{\frac{d}{I_{z}}}U_{4}, 
			\end{aligned}
		\end{equation}
		\begin{equation}
			\begin{aligned}
				\dot{\phi} &= p+\sin(\phi)\tan(\theta)q+\cos(\phi)\tan(\theta)r, \\
				\dot{\theta} &= \cos(\phi)q-\sin(\phi)r, \\
				\dot{\psi} &= {\frac{\sin(\phi)}{\cos(\theta)}}q+{\frac{\cos(\phi)}{\cos(\theta)}}r,
			\end{aligned}
		\end{equation}
	\end{subequations}
	where $(x,y,z)$~[m] denote the position of the quadrotor center of mass with respect to the earth inertial reference frame, $(\phi,\theta,\psi)$~[deg] are the Euler angles indicating the orientation of the quadcopter with respect to the same reference frame, and $(p,q,r)$~[deg/s] are the associated attitude velocities. By denoting the motors angular rates as $\Omega_{i}$, $i=1,\ldots,4$, the four inputs in \eqref{eq:quadrotormodel} are defined so as to be linear in the control variables, \emph{i.e.,}
	\begin{equation}\label{eq:quad_inputs}
		\begin{aligned}
			U_{1} &= b\sum_{i=1}^{4}\Omega_{i}^{2},\\ 
			U_{2} &= bl(\Omega_{4}^{2}-\Omega_{2}^{2}),\\ 
			U_{3} &= bl(\Omega_{3}^{2}-\Omega_{1}^{2}),\\
			U_{4} &= d(-\Omega_{1}^{2}+\Omega_{2}^{2}-\Omega_{3}^{2}+\Omega_{4}^{2}),
		\end{aligned}
	\end{equation}
	while $\Omega_{R}=-\Omega_{1}+\Omega_{2}-\Omega_{3}+\Omega_{4}$ and the remaining parameters are reported, along with their physical meaning, in \tablename{~\ref{tab:2}}.
	
	The data-collection phase is carried out in closed-loop, by controlling the position and attitude of the quadrotor with the controller proposed in \cite{Formentin2011} and three \emph{proportional derivative} (PD) controllers, respectively, \emph{i.e.,}
	\begin{subequations}
		\label{eq:PIDcontrolquadrotor}
		\begin{align}        
			U_{1} &= \frac{mg-k_1(z-\hat{z})-k_2\dot{z}}{\cos{\phi}\cos{\theta}},\label{eq:flatnessbased}\\ 
			U_{2} &= -I_{x}(k_3(\phi-\hat{\phi})+k_4\dot{\phi}),\label{eq:PD1}\\ 
			U_{3} &= -I_{y}(k_5(\theta-\hat{\theta})+k_6\dot{\theta}),\label{eq:PD2}\\
			U_{4} &= -I_{z}(k_7(\psi-\hat{\psi})+k_8\dot{\psi})\label{eq:PD3}.
		\end{align}
	\end{subequations}
	whose parameters are reported in \tablename{~\ref{tab:3}}. By using these controllers, which are assumed to be unknown throughout the design of the predictive controller, we perform $L=10$ experiments of length $10$~s, corresponding to subsets comprising $T=400$ samples, since the sampling time is set to $T_{s}=0.025$~s. For the system to lie within the framework considered in the paper, the PD controllers are exploited to keep $\theta$ and $\phi$ close to zero both in the data-collection\footnote{When gathering data, the set point for both angles are selected as slowly varying signals, randomly generated within $[-0.2,0.2]$ [deg]. This choice allows us to retain information on possible configurations in which the attitude angles are not exactly zero.} and testing phases, thus allowing us to decouple the altitude dynamics from the one of the other variables in \eqref{eq:quadrotormodel} and to set the problem into the framework considered in the paper. Instead, since we aim at exploiting the E-DDPC to replace the controller in \eqref{eq:flatnessbased}, the closed-loop is fed with a piecewise-constant reference for the altitude, which is randomly generated within $[0,4]$~[m] to guarantee that the altitude dynamics is persistently excited. Both the height $z$~[m] and the vertical velocity $\dot{z}$~[m/s] of the quadrotor are assumed to be measured, with the available measurements corrupted by white zero-mean noise yielding an average SNR approximately equal to $35$~[dB] over the two channels and experiments. 
	
	Based on our choices, we can focus on designing the E-DDPC attitude controller via the available noisy measurements only, which are averaged prior to the actual control design phase to exploit the strategy presented in Section~\ref{sec:handlenoise}. The design is performed by pre-compensating the gravitational force, namely we introduce
	\begin{equation}\label{eq:precomp}
		u_{1}=\frac{U_{1}}{m}-g,
	\end{equation}   
	which is the actual variable of the predictive control problem. By considering the tracking formulation in \eqref{eq:trackingDDPC}, E-DDPC is retrieved for $Q=diag([1,0])$, thus solely penalizing altitude tracking, $R=0.01$, $P=diag([100,100])$, and setting $N=5$. To avoid crashes, the altitude is forced to be grater than or equal to zero, while the input $u_{1}$ is forced to lie within the following interval
	\begin{equation}\label{eq:constrQUAD}
		-9.81 \leq u_{1} \leq 9.564,
	\end{equation} 
	where the lower bound correspond to a null action $U_{1}$ and the upper bound is dictated by the features of the motors used to control the quadcopter.
	
	The obtained E-DDPC law is then tested in both take-off and landing maneuvers, while the attitude PD controllers in \eqref{eq:PD1}-\eqref{eq:PD3} are used to track zero roll and pitch references. Specifically, we use the explicit controller to bring the quadcopter to a cruise altitude of $1$~[m] in the first case, while we exploit it to return to the ground from such an altitude in the second case. \figurename{~\ref{Fig:noisy_testQUAD}} shows the results of the two test, performed when the measured altitude and vertical velocity are corrupted by noise with the same intensity considered in the data generation phase. Despite the noise acting on the data that guide the selected control action at each time step, both maneuvers are successfully performed, as proven by the actual trajectories of the quadcopter.          
	
	\subsubsection{Sensitivity to the tuning parameters}
	\begin{table}[!tb]
		\caption{Altitude control example: performance indicators \emph{vs} values of $q_{1}$.}\label{tab:4}
		\centering
		\begin{tabular}{ccccc} 
			$q_{1}$ & $10^{-1}$ & 1 & $10$ & $10^{2}$\\
			\hline 
			${T}_{sett}$ [s] & 1.0 & 1.9 & 2.4 & 2.5 \\
			\hline
			$S_{max}$ [\%] & 1.6 & 0 & 0 & 0\\
			\hline
			$B_{\%}$ [\%] & 1.3 & 1.3 & 2.5 & 3\\
			\hline
		\end{tabular}
	\end{table}
	\begin{table}[!tb]
		\caption{Altitude control example: performance indicators \emph{vs} values of $R$.}\label{tab:5}
		\centering
		\begin{tabular}{cccc} 
			$R$ & $10^{-4}$ & $10^{-3}$ & $10^{-2}$\\
			\hline 
			$T_{sett}$ [s] & 1.4 & 2.4 & 2.4 \\
			\hline
			$S_{max}$ [\%] & 6.6 & 2.6 & 0 \\
			\hline
			$B_{\%}$ [\%] & 7.2 & 4.2 & 1.3 \\
			\hline
		\end{tabular}
	\end{table}
	\begin{figure}[!tb]
		\centering
		\includegraphics[scale=.7,trim=1.5cm .8cm 3cm 1cm,clip]{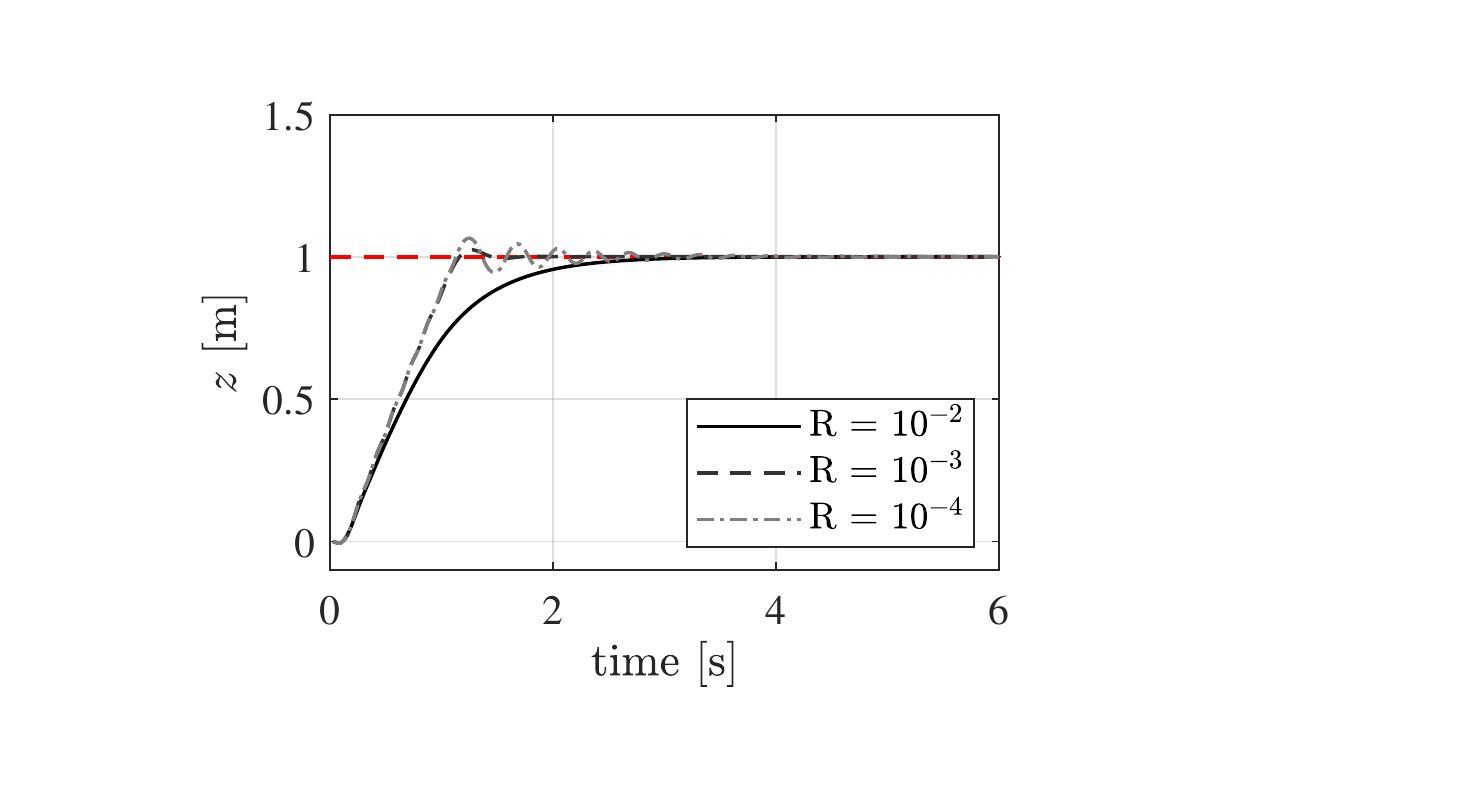}
		\caption{Altitude control example: quadcopter trajectory \emph{vs} $R$. The red dashed line indicates the cruise altitude.}\label{Fig:sensitivity}
	\end{figure}
	We now analyze how the performance of E-DDPC are shaped by different choices of the tunable weights $Q$ and $R$, still fixing the second element in the diagonal of $Q$ equal to zero. For the sake of clarity, we here consider noiseless take-off tests only, and we quantitatively assess the performance attained by the explicit altitude controller by looking at: $(i)$ the settling time $T_{sett}$ [s] at which the cruise altitude is reached; $(ii)$ the maximum overshoot $S_{max}$~[\%] with respect to the target altitude, and $(iii)$ the percentage $B_{\%}$~[\%] of instants over the test horizon for which the control bounds are hit.
	
	Let $q_{1}$ be the element in position $(1,1)$ in the penalty matrix $Q$. As shown in \tablename{~\ref{tab:4}}, when $R=0.01$, the higher the first component, the more the control action is prone to hit the operational bounds dictated by the constraints in \eqref{eq:constrQUAD}. Instead, the lower the weight on the tracking error, the prompter is the tracking, at the price of an overshoot in the transient. Note that the latter is yet rather small. As expected, in take-off, for $q_{1}=1$ and different weights $R$, the lower the weight, the more the bounds in \eqref{eq:constrQUAD} are hit (see \tablename{~\ref{tab:5}}). These saturations are paired with more consistent overshoots in altitude and oscillations around the target, as shown in \figurename{~\ref{Fig:sensitivity}}.  
	\subsubsection{Comparison with model-based solutions}
	\begin{figure}[!tb]
		\centering
		\includegraphics[scale=.7,trim=1cm 1cm 0cm 1cm,clip]{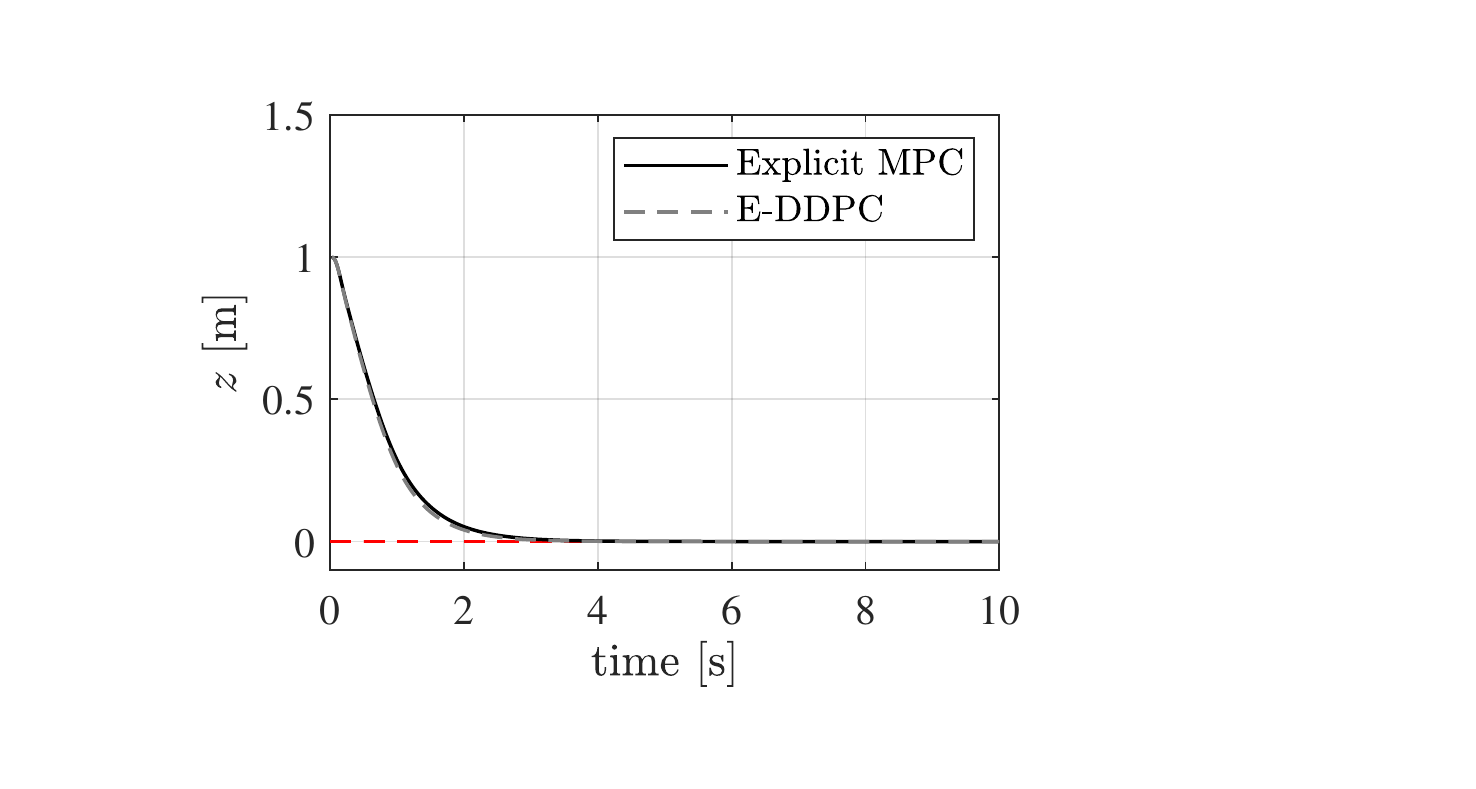}
		\caption{Altitude control example: E-DDPC \emph{vs} explicit MPC with the identified model. The difference between the two solutions is almost negligible over the horizon, except for the slight deviation of the quadrotor trajectory around $2$~s.}\label{fig:compMB}
	\end{figure}
	\begin{figure}[!tb]
		\centering
		\includegraphics[scale=.5,trim=2cm 9.25cm 0cm 10cm,clip]{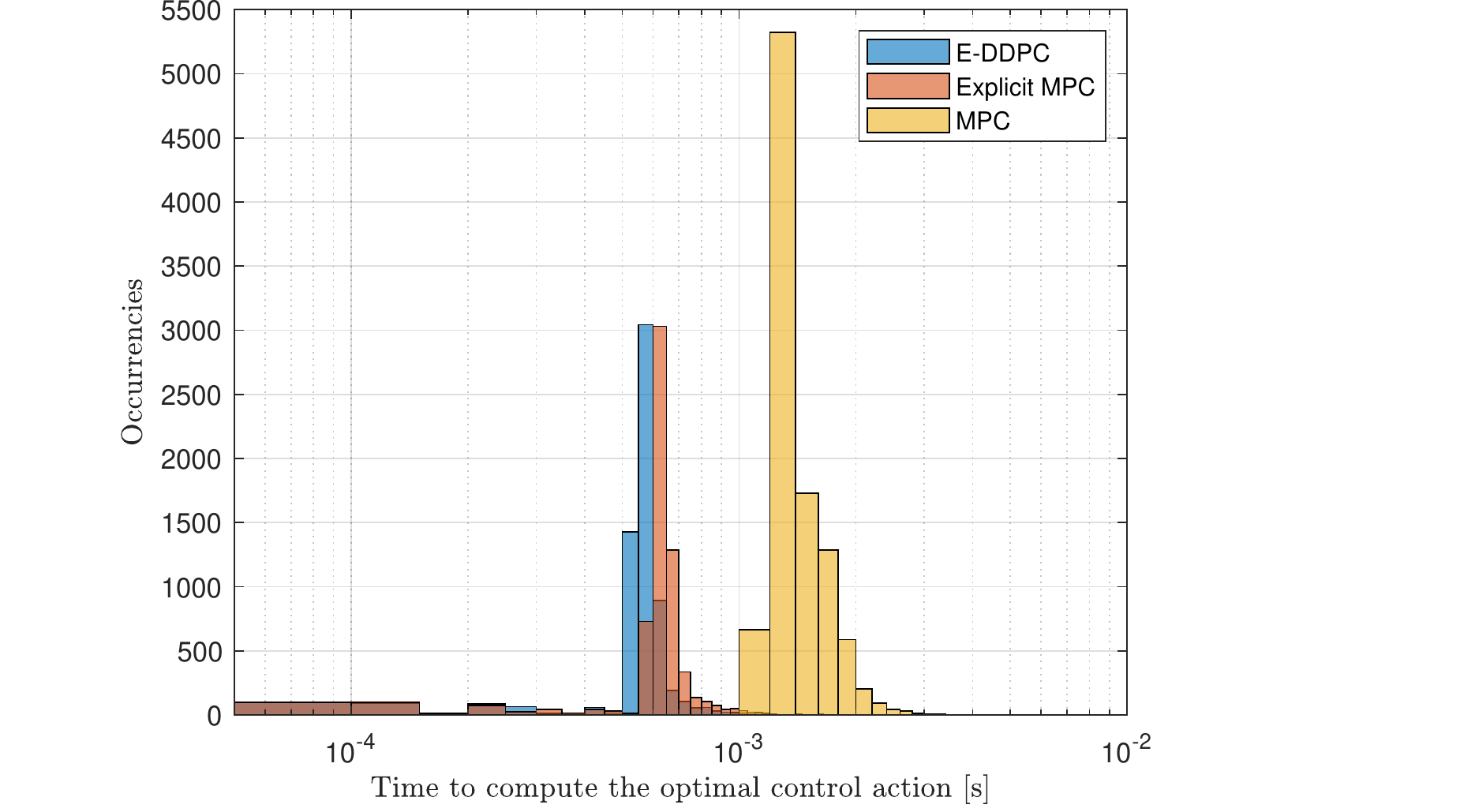}
		\caption{Altitude control example: distributions of the CPU times needed to retrieve the control action at each time step with E-DDPC, explicit and implicit MPC. As expected, the explicit solutions are more computationally efficient than the implicit one.}\label{fig:comparisonCPU}
	\end{figure}
	\begin{table}[!tb]
		\caption{Altitude control: E-DDPC \emph{vs} implicit and explicit MPC. CPU times, storage requirements and regions of the PWA laws.}\label{tab:compCPU}
		\centering
		\begin{tabular}{cccc}
			\multicolumn{1}{c}{} \hspace*{-.2cm}&\hspace*{-.2cm} Implicit MPC \hspace*{-.2cm}&\hspace*{-.2cm} Explicit MPC \hspace*{-.2cm}&\hspace*{-.2cm} E-DDPC\\
			\hline
			$\#$ regions & - & 723 & 736\\
			\hline 
			$\bar{T}$~[s] & 1.5 $\cdot 10^{-3}$ & 0.5 $\cdot 10^{-3}$ & 0.5 $\cdot 10^{-3}$\\
			\hline 
			$T_{wc}$~[s] & 82 $\cdot 10^{-3}$ & 1.3 $\cdot 10^{-3}$ & 1.4 $\cdot 10^{-3}$\\
			\hline 
			Storage [kB] & 1.4 & 570 & 586 \\
			\hline
		\end{tabular}
	\end{table}
	The performance attained with E-DDPC are eventually compared with the ones obtained by designing an explicit MPC, with a model retrieved with the n4SID method \cite{VanOverschee1994}. To train the model, we consider one dataset of length $T=4000$ samples, with the same features as the ones used to design the data-driven controller. This choice yields a fair comparison of the two control laws, since the dimensions of the dataset and their characteristics are the same. The comparative tests are performed within a noiseless scenario. As shown in \figurename{~\ref{fig:compMB}} for a landing test, both the controllers result into the same behavior of the quadcopter. We remark that this equivalent response comes at the price of an additional effort required when designing the explicit MPC law, due to the identification phase needed to retrieve the model of the quadcopter. 
	
	Differences between the two solutions arise when comparing them in terms of the time required to compute the control action to be applied at each instant. Indeed, as reported in \figurename{~\ref{fig:comparisonCPU}}, less exploration time is generally needed when considering E-DDPC over $100$ tests performed for randomly generated initial conditions. As expected, the time required for an implicit MPC (designed with the identified model) to compute the control action is generally higher, since the use of the implicit controller entails the solution of a QP at each time step. These conclusions are further supported by results reported in \tablename{~\ref{tab:compCPU}}, indicating that both the explicit solutions considerably reduce the computational load to determine the control action, at the price of an increase in the on-board memory required to store the explicit laws rather than the matrices of the implicit control problem only. Note that, when considering the explicit controllers, the control action is always found within the considered sampling time (see the average time $\bar{T}$~[s] and the worst case time $T_{wc}$~[s] reported in \tablename{~\ref{tab:compCPU}}). Instead, in some tests (not shown in \figurename{~\ref{fig:comparisonCPU}}), the time required for the implicit MPC to find the control action even exceeds the CPU time dictated by $T_{s}$.

	\section{Conclusions}\label{sec:conclusions}
	In this paper, we have derived an explicit data-driven predictive control (E-DDPC) solution. Our formulation relies on the Willems' fundamental lemma, leading to a fully data-driven piecewise affine law, that is designed so as to optimize a quadratic performance index and satisfy a set of user-defined constraints. To account for the pervasive presence of noise in real data, we further propose a noise handling strategy, the effectiveness of which has been assessed on three numerical examples.
	
	Future research will be devoted to generalize E-DDPC to a purely input/output setting, and to extend it to nonlinear systems. Moreover, alternative strategies to merge the regions of the explicit law and to manage noise will be investigated.
	
	\bibliographystyle{plain}
	\bibliography{EDDPC} 
	
	\appendix
	\section{Proof of Lemma~\ref{lemma:2}}\label{appendix1} 
	Let us recast the equation in \eqref{eq:MPCconstr1} as
	\begin{equation}\label{eq:state_compact} 
		x(k+1)=\mathcal{M}\begin{bmatrix}u(k)\\ x(k) \end{bmatrix},
	\end{equation}
	according to which the dynamics of the system over the prediction horizon can be compactly expressed as:
	\begin{equation}\label{eq:state_evol}
		X=\Xi x +\underbrace{\begin{bmatrix}\Gamma\\ \mathbf{0}_{n(N_{x}-N_{u}-1)\times N_{u}m}\end{bmatrix}}_{\bar{\Gamma}} U,
	\end{equation}
	where $X \in \mathbb{R}^{nN_{x}}$ stacks the predicted states $\{x(i)\}_{k=1}^{k=N_{x}}$, $\Gamma\in\mathbb{R}^{nN_{u}\times mN_{u}}$ is defined as
	\begin{equation*}
		\Gamma = 
		\begin{bmatrix}
			\gamma & \mathbf{0}_{n\times m} & \cdots & \cdots & \mathbf{0}_{n\times m} \\
			\xi \gamma & \gamma & \mathbf{0}_{n\times m} & \cdots & \mathbf{0}_{n\times m} \\
			\vdots & \vdots & \vdots & \ddots & \vdots \\
			\xi^{N_{u}-1} \gamma& \xi^{N_{u}-2} \gamma & \cdots & \cdots & \gamma 
		\end{bmatrix},
	\end{equation*}
	and
	$\Xi\in \mathbb{R}^{nN_{x}\times n}$ is
	\begin{equation}
		\Xi = \begin{bmatrix}
			\xi' &
			\ldots &
			\left(\xi^{N_{u}}\right)' & (\xi_{K}^{N_u+1})' & \ldots & (\xi_{K}^{N_x})'
		\end{bmatrix}',
	\end{equation}
	with
	\begin{align}\label{eq:help_matrices}
		\xi= &\mathcal{M}\begin{bmatrix}
			0_{m\times n} \\ I_n
		\end{bmatrix}\!\!,~~~\gamma=\mathcal{M}\begin{bmatrix}
			I_m\\0_{n\times m}
		\end{bmatrix},\\
		&~~~~~~~~~~~~~~\xi_{K}=\xi+K\gamma.
	\end{align}
	By exploiting \eqref{eq:state_compact}, it can also be shown that $H$ in \eqref{eq:QP2} is given by
	\begin{equation} \label{eq:HMB}
		H = \mathcal{R}+\bar{\Gamma}'\mathcal{Q}\bar{\Gamma},
	\end{equation}
	where $\mathcal{Q} \in \mathbb{R}^{nN_{x} \times nN_{x}}$ and $\mathcal{R} \in \mathbb{R}^{mN_{u} \times mN_{u}}$ are diagonal matrices defined as
	\begin{align*}
		& \mathcal{Q}\!=\!diag([Q,\cdots,Q,Q+K'RK,\cdots,Q+K'RK,P]),\\
		&\mathcal{R}\!=\!diag([R,\cdots,R]).
	\end{align*} These definition allows us to translate \eqref{eq:MPC} into \eqref{eq:QP2} by performing the manipulations summarized in Section~\ref{sec:fromItoE}. Indeed, since the assumption of Theorem~\ref{th:1} holds, by exploiting the data-driven representation of the system, \eqref{eq:state_evol} and \eqref{eq:HMB} can be directly translated into their fully data-driven counterparts. Specifically, we can define $\xi_d$ as in \eqref{eq:xid} and
	\begin{subequations}
		\begin{equation}\label{eq:gammad}
			\gamma_d=X_{1,T} \begin{bmatrix}U_{0,1,T} \\ \hline X_{0, T} \end{bmatrix} ^\dagger \begin{bmatrix} I_{m} \\ 0_{n\times m} \end{bmatrix},
		\end{equation}
	\end{subequations}
	which are equivalent to \eqref{eq:help_matrices} and \eqref{eq:HMB} due to the result of Theorem~\ref{th:1}. We can thus find the data-dependent equivalents of $\Xi$ and $\Gamma$, namely
	\begin{equation*}
		\Xi_d \!=\! \begin{bmatrix}
			\xi_d \\
			\vdots \\
			\xi_d^{N_{u}}\\
			\xi_{K,d}^{N_{u}+1}\\
			\vdots \\
			\xi_{K,d}^{N_{x}}
		\end{bmatrix}\!\!,~~
		\Gamma_d \!=\!\! 
		\begin{bmatrix}
			\gamma_d & \mathbf{0}_{n\times m} & \cdots & \cdots & \mathbf{0}_{n\times m} \\
			\xi_d \gamma_d & \gamma_d & \mathbf{0}_{n\times m} & \cdots & \mathbf{0}_{n\times m} \\
			\vdots & \vdots & \vdots & \ddots & \vdots \\
			\xi^{N_{u}-1}_d \gamma_d& \xi^{N_{u}-2}_d \gamma_d & \cdots & \cdots & \gamma_d
		\end{bmatrix}\!\!,
	\end{equation*}
	with $\xi_{K,d}=\xi_{d}+K\gamma_{d}$, which allow us to recast \eqref{eq:state_evol} in a data-driven fashion as
	\begin{equation}\label{eq:state_evolDD}
		X=\Xi_d x +\underbrace{\begin{bmatrix}\Gamma_d\\ \mathbf{0}_{n(N_{x}-N_{u}-1)\times N_{u}m}\end{bmatrix}}_{\bar{\Gamma}_{d}} U,
	\end{equation}
	and to find the data-driven counterpart of $H$ in \eqref{eq:HMB} as
	\begin{equation}
		H_d = \mathcal{R}+\bar{\Gamma}_d'\mathcal{Q}\bar{\Gamma}_d. \label{eq:Hd}
	\end{equation}
	These quantities ultimately allow us to retrieve the data-based predictive formulation in \eqref{eq:DDQP} through the same manipulations discussed in Section~\ref{sec:fromItoE}, whose equivalence with the model-based counterpart \eqref{eq:QP2} straightforwardly follows from the equivalence of \eqref{eq:state_evolDD} and \eqref{eq:state_evol}, and the one of \eqref{eq:Hd} and \eqref{eq:HMB}. 
	
\end{document}